%% file: main.tex
\documentclass{article}

\usepackage{macros}
\usepackage{thmtools} 
\usepackage{thm-restate}

\title{Optimizing the order of actions in contact tracing}

\author{
	Michela Meister \\ 
	meister@cs.cornell.edu
	\and
	Jon Kleinberg \\
	kleinberg@cornell.edu
	}

\begin{document}
\maketitle

\begin{abstract}
Contact tracing is a key tool for managing epidemic diseases like HIV, tuberculosis, and COVID-19. Manual investigations by human contact tracers remain a dominant way in which this is carried out. This process is limited by the number of contact tracers available, who are often overburdened during an outbreak or epidemic. As a result, a crucial decision in any contact tracing strategy is, given a set of contacts, which person should a tracer trace next? In this work, we develop a formal model that articulates these questions and provides a framework for comparing contact tracing strategies. Through analyzing our model, we give provably optimal prioritization policies via a clean connection to a tool from operations research called a ``branching bandit''. Examining these policies gives qualitative insight into trade-offs in contact tracing applications.
\end{abstract}

\input{sections/introduction}
\input{sections/example}
\input{sections/model}
\input{sections/results}
\input{sections/basic}
\input{sections/univariate}
\input{sections/bivariate}
\input{sections/general}
\input{sections/discussion}
\input{sections/acknowledgements}

\bibliographystyle{plain}
\bibliography{references}

\input{sections/appendix}

\end{document}

%% file: sections/introduction.tex
\section{Introduction} \label{section:introduction}
Contact tracing, the process of identifying individuals exposed to an infected case of a disease, is a key tool for managing epidemic diseases like HIV, tuberculosis, and COVID-19.~\cite{hiv_guide,tb_guide,vecino,keeling}. Once a contact is identified, they may undergo testing, quarantine, or medical treatment. Identifying infected contacts soon after exposure mitigates the impact of disease on multiple fronts: quarantining cases is crucial to limiting the disease's spread~\cite{macintyre,riley,kretzschmar,klinkenberg,who-em}, and early initiation of treatment is associated with positive health outcomes such as shorter hospital stays and lower mortality~\cite{may,krawczyk,sobrino,nakagawa,greenaway,alvarez,grinsztejn}.

Manual investigations by human contact tracers remain a dominant way in which contact tracing is carried out. Typically, a team of contact tracers work together to interview infected cases and follow up with contacts, tasks which are often arduous and time-consuming~\cite{spencer,tb_guide}. To simplify things, we consider the workflow of a single \textit{tracer} tasked with investigating a set of contacts. The tracer iteratively chooses an individual from the set to \textit{query}. When an individual is queried, they are tested for infection; if they are infected, they receive medical treatment and list their contacts, which are then added to the set. Since the contact tracing process is limited by the number of tracers available, during an outbreak or epidemic there may not be resources to query each contact immediately. Therefore, an important strategic decision is, which contact should the tracer query next? For example, for an easily-transmissable respiratory disease like COVID-19, a tracer may prioritize querying a clerk at a grocery store over a writer working from home, as the clerk likely interacts with many more people each day than the writer. As different diseases have different vectors of transmission, these trade-offs change in other settings. 

Deciding which contact to query next quickly becomes complex, since each query a tracer makes has downstream effects. As illustrated in~\cref{fig:example_1}, suppose an infected case A exposed contacts B and C. If B became infected, they exposed D and E, and if C became infected they exposed F, G, and H.
\begin{figure}
\begin{center}
\includegraphics[width=.3\linewidth]{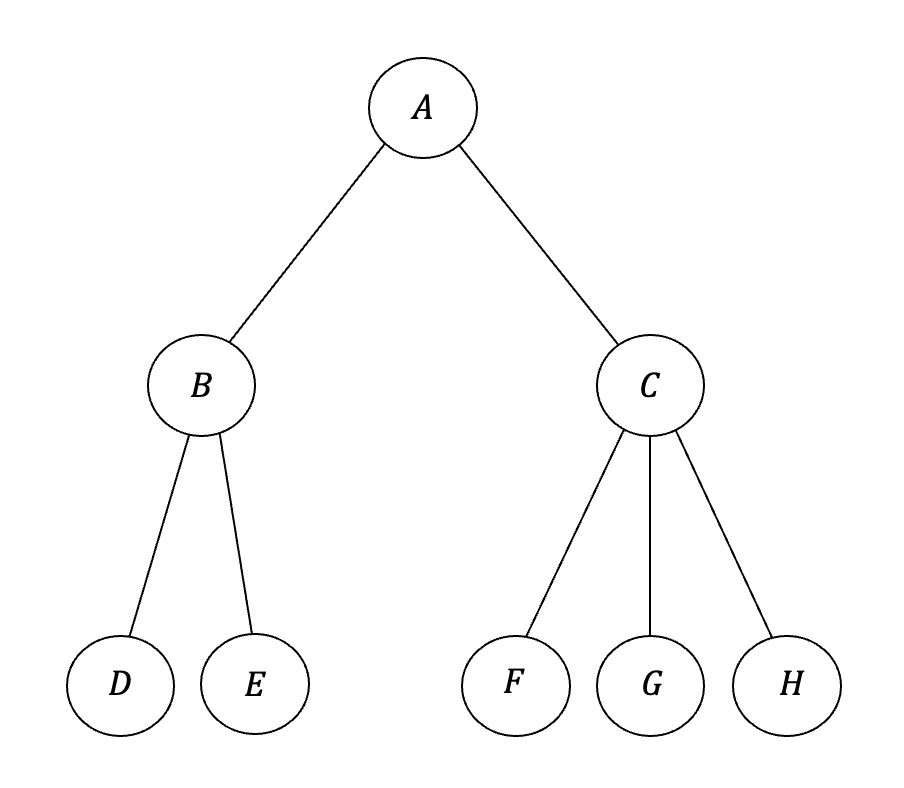}
\caption{Manual contact tracing begins with an index case (A), after which further contacts are revealed. Contact tracing proceeds recursively, where the tracer must decide the order in which to investigate contacts, with different orderings yielding different rewards.}
\label{fig:example_1}
\end{center}
\end{figure}
Upon querying A, the tracer only has access to first-generation contacts B and C and must decide which one to query next. Querying B opens up the possibility of querying second-generation contacts D and E, whereas querying C gives access to F, G, and H. Thus the decision whether to query B or C first affects the options available going forward.

Many other factors, such as the probability a contact is infected, their risk of infecting others, and the recency of their exposure affect these choices, and tracers often trade off between these different factors when deciding which contact to query next. To guide this decision-making, groups like the WHO and the CDC develop detailed recommendations for conducting contact tracing investigations~\cite{tb_guide,hiv_guide,who-em,virginia,northcarolina,cdc_covid_prioritize,cdc_covid_interview}. A primary aim of these guidelines is to synthesize these trade-offs into a decision-making protocol for tracers to follow. The complexity of this process is emphasized in the CDC's \textit{Guidelines for Investigations of Contacts with Infectious Tuberculosis}.
\begin{quote}
``Contact investigations are complicated undertakings that typically require hundreds of interdependent decisions, the majority of which are made on the basis of incomplete data, and dozens of time-consuming interventions.''~\cite{tb_guide}
\end{quote} 
Protocols for prioritizing contacts typically take the form of a flow-chart or matrix that assigns contacts to different groups and then dictates the order in which these groups ought to be queried. Often multiple criteria are considered when categorizing contacts. In many cases, two of the most important factors are the probability that a contact is infected and the recency of their exposure.

Despite the fact that the importance of these factors is well-understood, agencies still have difficulty managing these trade-offs, which is not without cost. During a 2017 HIV outbreak in West Virginia, contact tracers in the West Virginia Department of Health and Human Resources were overwhelmed by the surge in cases. In such a situation, CDC guidelines recommend interviewing contacts associated with clusters of infections, who may be at high risk of infection~\cite{hiv_guide}. However, the contact tracers had no means of adjusting the order in which they investigated cases to respond to the outbreak.
\begin{quote}
``[T]here was no supervisory triage system to respond to a cluster of HIV infections. As a result, [contact tracers] would investigate cases linearly --- prioritizing index case investigations over investigating contacts within clusters --- and did not have the flexibility to shift their priorities based on the identification of an ongoing HIV cluster.''~\cite{quilter}
\end{quote}
Given the importance of identifying infected cases quickly, both with respect to  preventing future infections and initiating medical treatment early, delays like this may cause real harm.

Even in the numerous cases where an agency has an effective method of prioritization, there are subtleties to how these priorities are chosen~\cite{virginia,cdc_covid_prioritize}. For example, the protocol for COVID-19 contact tracing developed by the North Carolina Department of Health and Human Services prioritizes contacts by recency of exposure, with one exception: contacts associated with a cluster or outbreak of infections are prioritized first, regardless when they were exposed~\cite{northcarolina}. By clearly dictating when to shift priorities in an investigation, this protocol addresses the issues presented in the West Virginia case. Yet many questions still remain. How are these priorities chosen? How do these tools translate to different settings? How might priorities change with slightly different factors or as parameters change? Is there even a way to ask these sorts of questions?

In this work, we develop a formal model that articulates these questions and provides a framework for comparing contact tracing strategies. Through analyzing our model, we give provably optimal prioritization policies via a clean connection to a tool from operations research called a ``branching bandit''~\cite{weiss}. Examining these policies gives qualitative insight into trade-offs in contact tracing applications. The model we study has two phases: first an infection spreads within a population; then the infection process halts and a contact tracing intervention begins. Of course, contact tracing is an extremely complex process requiring nuance and domain expertise, and there are many factors which this stylized model cannot capture --- for example, the dynamics of contact tracing while a disease is still spreading, which we explore in forthcoming work~\cite{dynamic}. Yet this two-phase model already exposes trade-offs and questions about prioritization, decision-making, and resource allocation, which not only seem important in their own right, but also seem prerequisite to understanding more complex settings. 

\paragraph{Paper Organization.}
First, in~\cref{sub:related_work} we discuss relevant work from the contact tracing, operations research, and computer science literatures. Then in~\cref{section:example} we present an example to illustrate the trade-offs at play in contact tracing. In~\cref{section:model} we present our formal model, and in~\cref{section:results} we give an overview of our results for the four models we consider. The remaining~\cref{section:basic,section:univariate,section:bivariate,section:general} present each of the four models, their optimal policies, and the analysis of these policies. Finally we discuss a few compelling directions for future work in~\cref{section:discussion}.

\subsection{Further Related Work} \label{sub:related_work}
Work closest to ours focuses on comparing contact tracing policies and the question of ``who to trace''. Specifically, Armbruster and Brandeau consider a network model where a tracer with limited capacity must decide at each step which contact to trace next. In~\cite{armbruster-brandeau_sim} via simulation they evaluate three different policies for prioritizing contacts. Using the best of these three policies, in~\cite{armbruster-brandeau}, they describe the trade-offs between investing in such a contact tracing effort versus directing funding toward other interventions. Tian et al. also consider a network model where they evaluate a set of contact tracing strategies targeted at various subgroups within a population~\cite{tian}. Among compartmental models, Hethcote et al. evaluate sets of targeted contact tracing strategies in~\cite{hethcote-yorke} and Eames considers targeted tracing for different population structures in~\cite{eames}. Our results differ from this prior work in that we provide provably optimal contact tracing policies, as opposed to evaluating a set of specified policies or strategies.

The importance of developing prioritization strategies for contact tracing under resource constraints is highlighted in recent surveys~\cite{muller-kretzschmar,kwok}. Kaplan et al. consider tracing under resource constraints as well, however in a somewhat different setting~\cite{kaplan_emergency,kaplan_analyzing}. There have also been many studies evaluating under what conditions contact tracing is effective and when a disease can be controlled via tracing~\cite{fraser,hellewell,klinkenberg,eames-keeling,keeling}. In~\cite{muller}, M{\"u}ller et al. analyze a branching process model and compare the fraction of contacts traced to the effective reproduction number. 

Digital contact tracing is surveyed in~\cite{ahmed}. Our work focuses on contact tracing carried out by human tracers and is orthogonal to digital contact tracing apps. In the digital setting, Lunz et al. develop optimal policies for deciding how to quarantine contacts of infected individuals, where some fraction of contacts can be traced by digital means~\cite{lunz}.

Within the operations research and computer science literatures, our work is most closely related to search problems on trees that do not have clear connections to contact tracing, but which use similar techniques. The problem closest to our work is the tree-constrained Pandora's Box problem studied in~\cite{boodaghians}, which also analyzes stochastic selection on a tree, but which involves a fairly different objective. Another related problem is stochastic probing with constraints~\cite{gupta}. While our model of contact tracing involves a tracer operating on a tree, it is quite different from the minimum latency problem on weighted trees~\cite{sitter} in that the tracer does not physically travel to the the individuals they trace. Finally, our model formulation can be viewed as falling under the general class of Markov decision processes, but the key is that it falls under the specific class of branching bandit problems, which leads to an efficient solution.

%% file: sections/example.tex
\section{Example} \label{section:example}
\begin{figure}[t]
\begin{center}
\includegraphics[width=\linewidth]{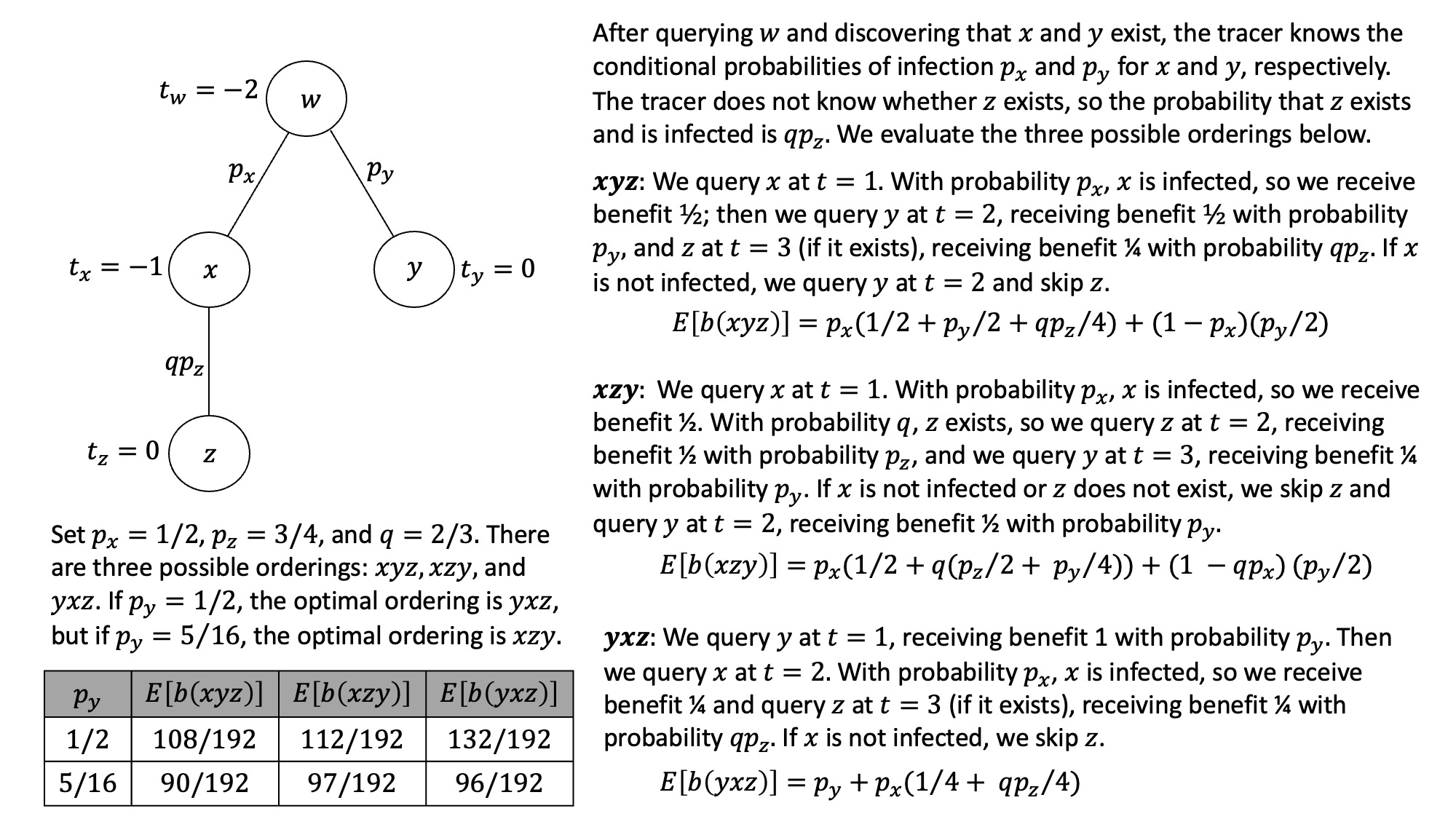}
\caption{In a mathematical model of contact tracing, the tracer knows the probability of infection for each contact and needs to choose an ordering for investigating contacts. As we show in this paper's main results, there is an efficient algorithm to decide an optimal ordering.}
\label{fig:example_diagram}
\end{center}
\end{figure}
We begin with an example that illustrates the trade-offs at play in prioritizing contacts. Suppose an infection spreads in a population over the course of two days,\footnote{We say ``day'' for simplicity; a day represents an arbitrary unit of time.} after which the infection halts and contact tracing begins. Specifically, over the course of two days each individual in a population meets one new contact each day with probability $q \in (0, 1]$, and infected individuals probabilistically infect each new person they meet. The following events take place:
\begin{itemize}
\item On day $t=-2$, $w$ is infected with probability $p_w$.
\item On day $t=-1$, with probability $q$, $w$ meets $x$ and if $w$ is infected they infect $x$ with probability $p_x$.
\item On day $t=0$, with probability $q$, $w$ meets $y$ and if $w$ is infected they infect $y$ with probability $p_y$.
\item On day $t=0$, with probabability $q$, $x$ meets $z$ and if $x$ is infected they infect $z$ with probability $p_z$. 
\item After day $t = 0$, the community goes into lockdown and no new infections occur.
\end{itemize}
Now consider a contact tracer working to discover infected cases. On day $t = 0$ the tracer finds that $w$ is infected and learns that they met $x$ on day $t = -1$ and also met $y$ on day $t = 0$. 

Going forward, on each day $t \geq 1$ the tracer chooses one individual to \textit{query}. When an infected individual is queried, they receive medical treatment and their contacts become available to query in the future. Querying an individual infected for $\tau$ days returns benefit $2^{-\tau + 1}$, which represents the probability they respond to treatment.\footnote{We will have more to say about how benefit depends on time, but roughly speaking this says that treating someone sooner is better than later.} Querying an uninfected individual returns benefit $0$. Each individual may be queried at most once. The tracer's objective is to maximize the total benefit accumulated. Since $x$ was exposed to $w$ on day $t_x = -1$, the tracer knows that $x$ may have met another contact (arbitrarily called $z$) on day $t = 0$, however they have no way of accessing $z$ (or even knowing if $z$ exists) before querying $x$. Since $y$ was exposed on day $t_y = 0$, the tracer knows that $y$ met no further contacts. 

The question is, which contact should the tracer query first, $x$ or $y$? Querying $x$ potentially grants access to $z$, however $y$ was exposed more recently than $x$, so if $y$ is infected it returns a higher benefit.~\Cref{fig:example_diagram} shows the benefit of different ways in which the tracer can operate, numerically evaluated for two different parameter settings.

This simple example already highlights a few crucial points we explore in this work. For one, querying the node with the higher expected immediate benefit is not always optimal. This lack of a simple rule might at first seem to imply that a contact tracer working in this synthetic model would need to calculate the expected benefit of each possible ordering in order to identify the optimal choice. 

In fact this is not the case. Via a dynamic programming approach we provide optimal policies for querying individuals that are computed before contact tracing begins and which are straightforward to implement: a ``priority index'' is computed for each individual, and at each step the contact tracer queries the individual with the highest priority index.

%% file: sections/model.tex
\section{Model} \label{section:model}
Modeling a contact tracing process involves a few different ingredients. People need to meet and make new contacts, through these interactions the infection needs to spread to some individuals and not others, and finally, we need a way of identifying infected individuals and their contacts. 

We develop a model with two phases. Phase 1 spans steps $-T \leq t \leq 0$ and involves a contact process, which describes how people meet new contacts, and an infection process, which describes how the infection spreads through these interactions. At the end of Phase 1, the contact and infection processes halt, and from then on no new infections occur. Phase 2 begins on step $t = 0$ and continues indefinitely. At the start of Phase 2 a set of \textit{index cases} are identified. We define an index case as an individual that is exposed to infection and becomes infected according to a probability function $p$. We are agnostic to the origin of the index cases; for example, an index case may have been identified via surveillance testing, another contact tracing effort, or random chance. Starting with the index cases, on each step $t \geq 0$ a contact tracer, simply called a \textit{tracer}, selects one individual to \textit{query}. Querying an individual models the traditional test-and-trace process: it reveals the individual's infection status, and if they are infected it reveals their contacts; these contacts may then be queried on future steps. When an infected individual is queried, they are also receive medical treatment for the disease. For the remainder of the second phase the tracer iteratively queries individuals with the goal of identifying infected cases as efficiently as possible.

\paragraph{Phase 1} Phase 1 spans steps $t = -T$ to $t = 0$. During this phase, individuals meet new contacts and the infection spreads through these interactions.

Let $D$ be an arbitary distribution on $\{ 0, 1, 2, \dots \}$. On each step each individual meets a random number of contacts $Z \sim D$, where $Z$ is drawn independently for each individual at each step. If an individual is infected, they infect each new contact they meet independently according to a probability function $p$ defined separately for each model we consider. We call these contacts \textit{exposed}. Exposed individuals are labeled by the recency of their exposure: an individual exposed at time $t = -h$ has \textit{recency} $h$. Since Phase 1 spans steps $t = -T$ to $t = 0$, exposed individuals have recencies in $\{ 0, 1, \dots, T \}$. We assume that an individual's recency can be observed but their infection status is hidden. 

\paragraph{Step $t = 0$:} At the end of step $t = 0$ the contact and infection processes halt, and from then on no new infections occur. 

To understand the system on step $t = 0$, consider an individual $v$ exposed on step $t = -h$ in Phase 1. If $v$ becomes infected, then for each step $h - 1 \leq t \leq 0$ in the remainder of Phase 1, $v$ exposes $Z_t \sim D$ individuals. Then by the end of step $t = 0$, $v$ has exposed a multiset of contacts $Z(h) = (Z_0, Z_1, \dots, Z_{h-1}) \sim D^h$ where $Z_j$ indicates the number of contacts of recency $j$. Thus we can model $v$ as the root of a tree, where the nodes in the first layer represent the contacts $v$ met after being exposed, the nodes in the second layer represent contacts individuals in the first layer met after meeting $v$, and so on. We call this a \textit{tree of potential exposures} because there is a path of contacts from $v$ to each individual in the tree along which the infection could potentially travel. Since the distribution $D$ on contacts is fixed, $v$'s recency $h(v)$ determines the distribution on the tree of potential exposures. We often refer to the nodes in the tree of exposures as $v$'s \textit{descendants}.

We call the probability that an exposed node is infected the \textit{probability of infection}. In the first model we examine, the probability of infection is constant. In the second model, the probability of infection depends on the node's recency as defined by a function $p(h)$. Either way, for both models a node's recency contains all the information needed to determine its probability of infection and the distribution on its tree of exposures. In the third model we consider, the probability of infection depends both on the node's recency and the recency of its parent. Going forward, we use the terms ``node'' and ``individual'' interchangeably.

\paragraph{Phase 2} Phase 2 begins on step $t = 0$ and continues indefinitely. Throughout Phase 2 an individual's infection status is fixed but hidden. During this phase a contact tracing effort proceeds.

Contact tracing begins on step $t = 0$ when a set of index cases are identified. From then on, at each step $t \geq 0$ the tracer selects one node to query. Querying a node reveals its infection status, and if it is infected reveals the node's children along with the recency of each child. These children may then be queried on future steps. We assume that nodes of the same recency are indistinguishable until they are queried. 

On step $t = 0$, the index cases are the only individuals available to query. The contact tracer observes the recency of each index case but has no information about any other individuals in the population. For the remainder of Phase 2, the tracer may only query a node that is an index case or the child of an infected node queried on an earlier step. Equivalently, we can view each index case as the root of a tree of potential exposures, which together form a forest. Through this query process the tracer maintains a sub-forest where any leaf not already queried is either a root or the child of an infected node. These leaves form the \textit{frontier}, and each step the tracer selects one node to query from the frontier. Observe that, by definition, each node in the frontier was exposed to infection in Phase 1. 

To understand this process further, consider the options available to the tracer each step. Since nodes of the same recency are indistinguishable until they are queried, the system on any step $t \geq 0$ is defined by a multiset $S_t = (X_0, X_1, \dots, X_T)$, where $X_j$ indicates the number of nodes of recency $j$ present in the frontier. We call $S_t$ the \textit{state} on step $t$. We use both notions of a multiset when referring to states; i.e. $S_t$ can be viewed as a collection of elements or as a vector of counts.

Querying an infected node $v$ returns a \textit{benefit}, which represents the probability that $v$ responds to medical treatment. The benefit of querying an infected node decays relative to the duration of the infection and depends on the node's recency $h(v)$ and the step $t$ it is queried, as defined by a function $b(h(v), t)$. The benefit of querying an uninfected node is 0, and each node may be queried at most once. The tracer's objective is to maximize the total expected benefit returned over the course of Phase 2.

\paragraph{Defining the objective.} 
On each step $t \geq 0$, the tracer selects a node $v_t$ to query. Since the tracer only selects nodes from the frontier, by definition the node $v_t$ was exposed to infection in Phase 1. Let $\mathbbm{1}(v_t)$ indicate whether $v_t$ is infected $(1)$ or uninfected $(0)$. If $v_t$ is infected, then benefit $b(h(v_t), t)$ is returned. Thus the total benefit the tracer accumulates over the course of Phase 2 is
\begin{align*}
\sum_{t \geq 0} \mathbbm{1}(v_t) \cdot b(h(v_t), t).
\end{align*}
The main objective is to develop a policy for querying nodes that maximizes the total expected benefit, where the expectation is taken over all realizations $\{ \mathbbm{1}(v_0), \mathbbm{1}(v_1), \mathbbm{1}(v_2), \dots \}$. Such a policy is called an \textit{optimal policy}. In order to make this problem tractable, like many other stochastic models we assume exponential discounting. Specifically, the benefit of querying a node infected for $\tau$ steps is $e^{-\beta \tau}$ for a fixed parameter $\beta > 0$.\footnote{That is, discounting begins at $t = 0$. For the sake of simplicity, for the example in~\cref{section:example} discounting begins at $t = 1$.}

%% file: sections/results.tex
\section{Overview of results} \label{section:results}
Our results can be summarized in three main contributions. 

\paragraph{1) Constructing optimal policies.} We show how to construct an optimal policy for any instance of our model. These policies have a special property: for any instance of our model the optimal policy can be described by an algorithm that assigns each node a ``type'', computes an index based on the type, and chooses the node of the highest index. Such a policy is called an \textit{index policy}.\footnote{The terms ``index'' and ``index policy'' are from the operations research literature and have no relation to the term ``index case''.} An index policy is efficiently computable if each individual index can be computed in polynomial time. For any instance of our model, we show how to compute an optimal policy that is an efficiently computable index policy. We prove this result in~\cref{section:general}.

We can interpret this result in the context of the contact tracing protocols discussed in~\cref{section:introduction}. Recall that many contact tracing protocols assign individuals to groups and then dictate the order in which groups ought to be queried. Since our construction computes indices from types, the resulting index policy induces a fixed priority ordering on types. Mapping individuals to nodes and groups to types, our results imply that, for our model, (1) the optimal policy overall is defined by a priority ordering on groups, and (2) this priority ordering has an explicit, efficient construction. It is important to note that, taken on its own, this result does not explicitly describe any policies, and it makes no guarantees about the structure of an optimal policy beyond the promise that it is an index policy. Describing the structure of optimal policies any further requires analyzing the construction itself. 

\paragraph{2) Analyzing optimal policies for different models of infection.} We examine three different versions of our model, each corresponding to a different model of infection. Analyzing the construction of optimal policies in each model gives qualitative insight into questions about prioritizing contacts from~\cref{section:introduction}, such as how to trade-off between an individual's probability of infection, the recency of their exposure, or the number of other contacts they may have exposed. The three versions we examine are identical except for a function defining the probability of infection. First we examine a basic model, where the probability of infection is constant, then we examine a univariate model, where the probability of infection decays absolutely with time, and finally we examine a bivariate model where a node's probability of infection decays relative to the incubation period of its parent.

\paragraph{3) Connecting contact tracing and the branching bandit problem.} Our key finding is a clean connection between contact tracing and the branching bandit problem~\cite{weiss}. The branching bandit problem broadly belongs to a large class of online decision problems called ``bandit'' problems. The general motivation for the branching bandit problem is scheduling projects, where each project returns a reward and begets new projects, which must then be scheduled. While bandit problems in general have numerous applications~\cite{gittins_text,slivkins_text,lattimore_text}, most applications of the branching bandit model are similar scheduling problems. We therefore find this connection between contact tracing and the branching bandit problem especially striking, since it extends the branching bandit problem to a domain to which it previously has not been applied. We formalize this connection in~\cref{section:general}.

\paragraph{Section Organization.} The remainder of this section reviews our results in the order in which they appear. First we discuss the three models of infection we examine: a basic model (\cref{section:basic}), a univariate model (\cref{section:univariate}), and a bivariate model (\cref{section:bivariate}). Then we discuss a general model of infection which encompasses the previous three models and demonstrates the connection to the branching bandit problem (\cref{section:general}). Finally, we give an overview of techniques. Throughout the paper, for simplicity we refer to ``the'' optimal policy when we are discussing a specific optimal policy we are constructing, however models may have multiple optimal policies. 

\subsection{Basic Model}
The basic model describes a standard model of infection where the probability of infection is constant. 

\begin{restatable*}{thm}{basicmain} 
\label{thm:basic_main}
In the basic model, it is optimal to query nodes in order of recency. 
\end{restatable*}

This result provides insight into some of the most vexing questions from~\cref{section:introduction}. Namely, in practice tracers seem to have these two opposing priorities --- whether to query more recent cases or cases that may have exposed many other contacts. In practice it seems that tracers lean towards querying in order of recency. In our model, this result shows that --- even taking into account the downstream effects of accessing an individual's contacts --- the optimal policy is indeed to query the most recent case first.

\subsection{Univariate Model}
In the univariate model the probability of infection decays absolutely with time according to an exponential functional form. The rate of decay is parameterized by a constant $\alpha \geq 0$, where for small values of $\alpha$, the probability of infection is close to constant, and the rate of decay accelerates as $\alpha$ increases. While in the basic model querying nodes in order of recency is optimal, this is not always the case in the univariate model. 

For small values of $\alpha$, when the probability of infection is close to constant and the setting resembles that of the basic model, the optimal policy queries nodes in order of recency. However, once $\alpha$ reaches a certain threshold, the optimal policy no longer queries nodes in order of recency, and one may wonder whether any structure remains at all. In fact, we find that there is still structure to the optimal policy: the policy always queries either the most recent or least recent node available. We say that such a policy is defined by an \textit{interleaved} priority ordering.
\begin{restatable*}[Interleaving property]{definition}{interleavingdef}
An ordering $\sigma$ on $\{ 0, 1, \dots, T \}$ is \textit{interleaved} if for all $0 \leq j \leq T$, $\sigma_j$ is either the maximum or minimum element of the suffix $\sigma_j, \dots, \sigma_T$.
\end{restatable*} 
Observe that many different priority orderings satisfy the interleaving property. For example, an ordering that prioritizes nodes by recency is interleaved, as is an ordering that prioritizes nodes by reverse recency. Our main result shows that the interleaving property holds for optimal policies in the univariate model.
\begin{restatable*}{thm}{univmain}
\label{thm:univ_main}
In the univariate model there is an optimal policy defined by an interleaved priority ordering.
\end{restatable*}
This result implies that the optimal policy always queries either the most recent node or least recent node in the frontier, which exposes an interesting trade-off between the probability that a node is infected, the recency of its exposure, and its expected number of children. Since a less recent node was exposed earlier in time, it has a larger number of children in expectation. Additionally, since the probability of infection decays with time, a less recent node also has a higher probability of infection. However, a more recent node, should it be infected, returns a higher immediate benefit. This result implies that the optimal policy pursues the extremes: it either queries the least recent node with the highest probability of infection and the most children in expectation, or it queries the most recent node, which is associated with the highest benefit. It is particularly interesting that the optimal policy never tends toward a node of intermediate recency, which seems to imply that compromising between these two extremes is suboptimal.

\subsection{Bivariate Model}
In the bivariate model the probability that a node is infected decays relative to the incubation period of its parent. As a result, defining a node in this model requires examining two parameters, the recency of the node and the recency of its parent. For a node of recency $h$ with a parent of recency $h'$, we call $\Delta = h' - h$ the \textit{span}, representing the span of time the parent has been infected upon meeting the child. The probability that a node is infected decays exponentially as a function of $\Delta$. In this model each node in the frontier is defined by a \textit{type} $(h, \Delta)$, and we examine policies on the set of types $\{ 0, 1, \dots, T \}^2$.
Thus the bivariate model demonstrates that we can analyze policies that take into account multiple parameters.

Analyzing optimal policies in the bivariate model reveals monotonic structure in the ordering of nodes with the same recency.
\begin{restatable*}{thm}{bivspan}
\label{thm:biv:span}
In the bivariate model, there is an optimal policy that queries nodes with the same recency in order of increasing span.
\end{restatable*}
We also explore monotonicity with respect to recency, in a more restricted model. Here we find that, for recencies in $\{ 0, 1, 2 \}$ and within certain constraints, nodes of the same span are queried in order of recency. 
\begin{restatable*}{thm}{bivrecency}
\label{thm:biv:recency}
In the bivariate model, for any Bernoulli distribution $D$, a large enough constant $\beta > 0$, and restricted to types with recencies in $\{ 0, 1, 2 \}$, it is optimal to query nodes of the same span in order of recency.
\end{restatable*}
While there are examples to show that a restriction on $\beta$ is necessary, the restriction to Bernoulli distributions and to recencies in $\{ 0, 1, 2 \}$ are functions of our proof technique.

\subsection{General Model}
The general model provides a broad framework that allows for a variety of factors to affect how individuals interact and how the infection spreads. As we saw in~\cref{section:introduction}, in practice individuals are often categorized according to multiple attributes, such as their profession or role within a community, their age, or the recency of their exposure. The general model captures this complexity, by assigning each node a type representing a set of attributes. Whereas types in the bivariate model are defined by two parameters, types in the general model can be defined by an arbitrary number of parameters. As a result, the general model encompasses the previous three models as well: in the basic and univariate models a node's type is its recency, and, as before, in the bivariate model a node's type is the pair $(h, \Delta)$.

Our main result shows that optimal policies in the general model are index policies on types. 
\begin{restatable*}{thm}{generalpolicy}
\label{thm:general}
For any instance of the general model, there is an optimal policy that is an index policy on the set of types. Moreover, this index policy has an efficient construction.
\end{restatable*}
This result implies that any instance of the general model has an optimal policy defined by a priority ordering on types. In order to prove this result, we show that any instance of the general model maps to an instance of the branching bandit model.
\begin{restatable*}{thm}{generalreduction}
\label{thm:general_reduction}
The general model reduces to the branching bandit model.
\end{restatable*}
This reduction formalizes the connection between contact tracing in our model and the branching bandit problem, by showing that finding an optimal policy for any instance of the general model requires analyzing an optimal policy for a corresponding instance of the branching bandit model. 

\subsection{Summary of techniques}
Here we formally define the branching bandit problem and describe how the reduction from the general model to the branching bandit problem lays the foundation for our results in the basic, univariate, and bivariate models.

The branching bandit model involves arms belonging to classes $\{ 1, \dots, L \}$, where when an arm of class $i$ is pulled it yields a non-negative reward $R(i)$, occupies $\mu(i)$ steps, and is replaced by a set of new arms $N_{i1}, \dots, N_{iL}$. Each class $i$ has an arbitrary, known, joint distribution on the random variables $R(i)$, $\mu(i)$, and $N_{i1}, \dots, N_{iL}$. At each step $t$ the system is defined by a vector $n(t) = (N_1, \dots, N_L)$, where $N_i$ is the total number of arms of class $i$ available, and a reward received at step $t$ is discounted by $e^{-\eta t}$ for a fixed parameter $\eta > 0$. The objective is to find a policy for pulling arms that maximizes the total discounted reward accumulated. As we describe in~\cref{sub:basic:weiss}, the optimal policy in the branching bandit model is an index policy on the set of classes $\{ 1, \dots, L \}$ with an efficient construction.

In~\cref{section:general} we reduce the general model to the branching bandit model, under a mapping where nodes map to arms, types map to classes, benefit maps to reward, and new children revealed by querying a node map to new arms acquired by pulling an arm. This reduction implies that the optimal policy for any instance of the general model is an index policy on types. Since the basic, univariate, and bivariate models are all versions of the general model, this guarantee extends to these three models as well. In particular, for each instance the reduction provides a construction for a priority ordering that defines the optimal policy. However, the construction on its own does not describe the policy explicitly or define any properties of the policy beyond the promise that it is an index policy. Revealing the structure of these policies requires analyzing the construction of optimal policies for each model, which is the main focus of the following sections. 

%% file: sections/basic.tex
\section{Basic Model} \label{section:basic}
The basic model describes a standard model of infection where the probability of infection is a constant $p_T \in (0, 1]$. Since the benefit of querying an individual infected for $\tau$ steps is $e^{-\beta \tau}$, the benefit of querying an infected individual of recency $h$ on step $t$ is $b(h, t) = e^{-\beta (h + t)}$. The basic model is a special case of both the univariate and bivariate models.

\paragraph{Understanding the trade-offs at play.} 
As described in~\cref{section:model}, we can view each node in the frontier as the root of a tree of exposures, where less recent nodes have more children in expectation. Recall that if a node is queried and found to be infected, its children are added to the frontier. Therefore, querying a less recent node provides an opportunity to significantly expand the frontier. On the other hand, since the benefit of querying an infected node decays with time, querying a more recent node returns a higher expected benefit. When selecting a node to query, how should the tracer trade-off between these two factors --- on the one hand, the expected benefit associated with querying a node, and on the other hand, the opportunity to access its descendants in the future?

Our main result shows that --- even taking into account these downstream effects --- the optimal policy queries nodes in order of recency. 
\basicmain

One might assume that such a straightforward policy has a correspondingly straightforward proof. After all, since querying the most recent node returns the highest expected benefit, perhaps an elementary exchange argument is sufficient. 

\paragraph{An attempt at an elementary exchange argument.} Following this line of reasoning, suppose on step $t$ a node $u$ is the most recent node in the frontier, but the tracer selects a less recent node $v$ and queries $u$ on some later step. Now consider exchanging the order of $u$ and $v$ so that $u$ is queried on step $t$ instead. If $u$ is infected, then its children are added to the frontier. Since any child is more recent than its parent, and since $u$ is the most recent node in the frontier on step $t$, on step $t + 1$ the children of $u$ are more recent than all other nodes in the frontier. Therefore, a commitment to prioritizing nodes by recency requires querying the descendants of $u$ \textit{recursively} in order of recency.

Exchanging $u$ and $v$ then becomes tricky. As a thought experiment, consider querying either $u$ or $v$ in isolation and then recursively querying any descendants in order of recency. Since a node is always less recent than its descendants, querying $u$ only ever leads to querying nodes more recent than $u$. However, since $v$ is less recent than $u$, querying $v$ could lead to querying nodes more recent than $v$ but less recent than $u$. From this standpoint, it's unclear how to compare these two processes or go about an exchange. One observation is that we could potentially compare the two processes if we measured the process catalyzed by $v$ only up until the first step where a node less recent than $u$ is queried. In fact, such a scheme already exists; continuing with this argument (which is now far from elementary) essentially requires reinventing machinery developed by Weiss for the branching bandit model.

\subsection{Summarizing Weiss's Model} \label{sub:basic:weiss}
Here we summarize Weiss's branching bandit model and the optimal policy for pulling arms, with a slight derparture from the original notation in~\cite{weiss}. Then in~\cref{sub:basic:defining} we map the basic model to the branching bandit model by mapping each node to an arm and each recency to a class of arms. 

Recall that the branching bandit model is a general framework involving arms of classes $\{ 1, 2, \dots, L \}$, where each time an arm of class $i$ is pulled it returns a non-negative reward $R(i)$, occupies $\mu(i)$ steps, and is replaced by a set of new arms $N_{i1}, \dots, N_{iL}$. Weiss's key idea is the notion of a \textit{period}. A period is defined with respect to any arbitrary priority ordering $\sigma$ on classes $\{ 1, 2, \dots, L \}$. For any class $i \in \{ 1, 2, \dots, L \}$, an $(i, \sigma_j)$-period is defined as follows. Initially only a single arm of class $i$ exists in the system. On step $t=0$ the arm of class $i$ is pulled, and from then on at each step an arm is pulled according to the priority ordering $\sigma$ until all classes $i' \preceq \sigma_j$ are exhausted. Therefore, an $(i, \sigma_j)$-period is really defined with respect to the prefix $\sigma_0, \sigma_1, \dots, \sigma_j$, since the ordering of the later elements is irrelevant. 

A few random variables describe an $(i, \sigma_j)$-period. Let $r(i, \sigma_j)$ be the total discounted reward accumulated during the period, and let $\tau(i, \sigma_j)$ be the duration. Observe that following a period of duration $\tau(i, \sigma_j)$, the reward of the next query is pre-multiplied by the discount factor $e^{-\eta \tau(i, \sigma_j)}$. Call $\gamma(h, \sigma_j) = \Exp[e^{-\eta \tau(i, \sigma_j)}]$ the expected pre-multiplier. 

\paragraph{Defining the optimal policy in Weiss's model.} Weiss leverages this notion of a period to inductively construct the optimal priority ordering. He then proves that the index policy defined by this optimal priority ordering is in fact an optimal policy outright.

The optimal priority ordering is constructed via a dynamic program which maintains an optimal prefix that lengthens in each round. In round $0$ the highest priority element $\sigma_0$ is selected. Entering any round $k > 0$, the prefix $\sigma_0, \sigma_1, \dots, \sigma_{k-1}$ is fixed, and $\sigma_k$ is selected by comparing $(i, \sigma_{k-1})$-periods over all classes $i$ \textit{not} in the prefix.

Describing the base case requires one more definition. For any $i \in \{ 1, 2, \dots, L \}$ an $(i, \emptyset)$-period simply pulls an arm of class $i$ exactly once. The reward $r(i, \emptyset)$, duration $\tau(i, \emptyset)$, and pre-multiplier $\gamma(i, \emptyset)$ are defined analogously, so by definition $\tau(i, \emptyset) = 1$ and $\gamma(i, \emptyset) = e^{-\eta}$. 

The highest priority is assigned to the element with the highest expected immediate reward. 
\begin{align*}
\sigma_0 &= \mathop{\arg\max}_{ i \in \{ 1, \dots, L \} } \Exp[r(i, \emptyset)]
\end{align*}
In any round $k > 0$, the prefix $\sigma_0, \sigma_1, \dots, \sigma_{k - 1}$ is fixed, and $\sigma_k$ is selected from the elements not already in the prefix.
\begin{align*}
\sigma_k &= \mathop{\arg\max}_{ i \in \{1, \dots, L \} \setminus \{ \sigma_0, \dots, \sigma_{k-1} \} } \frac{ \Exp[r(i, \sigma_{k-1})]}{1 - \gamma(i, \sigma_{k-1})}
\end{align*}
Section 3 of~\cite{weiss} proves that the priority ordering $\sigma$ constructed via this dynamic program is the optimal priority ordering, and moreover, that the index policy defined by $\sigma$ is the optimal policy overall.

\paragraph{An overview of Weiss’s proof idea.} For a high-level overview of the proof idea, first observe that $\sigma_0$ is the element associated with the maximal expected immediate reward. To understand the intuition behind this selection, imagine choosing between a node $u$ which returns the maximal expected immediate reward and some other node $v$. Since the expected immediate reward of any descendant of $v$ is at most that of $u$, there is no reason to delay querying $u$ in order to access the descendants of $v$. 

Moving ahead to any round $k > 0$, committing to the prefix $\sigma_0, \sigma_1, \dots, \sigma_{k - 1}$ implies that if we query a node, we are committing to recursively querying its descendants according to the ordering defined by the prefix. Therefore, we are no longer comparing individual queries but instead $(h, \sigma_{k - 1})$-periods. To compare periods, we need to trade-off between the expected reward $\Exp[r(h, \sigma_{k - 1})]$ returned and the expected pre-multiplier $\gamma(h, \sigma_{k - 1})$ imposed on the queries that follow. In this sense, we can think of $\sigma_k$ as selecting, from the elements not already in the prefix, the element $h$ whose $(h, \sigma_{k - 1})$-period has the highest expected “rate” of reward in this time-discounted setting.

\subsection{Defining the optimal policy} \label{sub:basic:defining}
Here we map the basic model of contact tracing to the branching bandit model by mapping nodes to arms, recencies to classes, benefit to reward, and the new children revealed by querying a node to new arms acquired by pulling an arm. We prove this reduction formally in~\cref{section:general}. We now restate the above dynamic program applied to the basic model, beginning with the definition of a period.

A period is defined with respect to any arbitrary priority ordering $\sigma$ on recencies $\{ 0, 1, \dots, T \}$. For any $h \in \{ 0, 1, \dots, T \}$, an $(h, \sigma_j)$-period is defined as follows. Initially only a single root $v$ of recency $h(v) = h$ exists in the frontier. On step $t=0$ the node $v$ is queried, and from then on at each step a descendant is queried according to $\sigma$ until all nodes $v'$ with recency $h(v') \preceq \sigma_j$ are exhausted. Let $b(h, \sigma_j)$ be the total discounted benefit accumulated during the period, let $\tau(h, \sigma_j)$ be the duration, and let $\gamma(h, \sigma_j) = \Exp[e^{-\beta \tau(h, \sigma_j)}]$ be the expected pre-multiplier. For any $h \in \{ 0, 1, \dots, T \}$ an $(h, \emptyset)$-period simply queries a single root $v$ of recency $h(v) = h$ exactly once. The benefit $b(h, \emptyset)$, duration $\tau(h, \emptyset)$, and pre-multiplier $\gamma(h, \emptyset)$ are defined analogously, so by definition $\tau(h, \emptyset) = 1$ and $\gamma(h, \emptyset) = e^{-\beta}$. Note that the benefit, duration, and pre-multiplier associated with a period are by definition non-negative.

Then $\sigma_0$ is assigned to the element with the highest expected immediate benefit.
\begin{align}
\sigma_0 &= \mathop{\arg\max}_{ h \in \{ 0, \dots, T \} } \Exp[b(h, \emptyset)] \label{eqn:dp_0}
\end{align}
In any round $k > 0$, the prefix $\sigma_0, \sigma_1, \dots, \sigma_{k - 1}$ is fixed, and $\sigma_k$ is selected from the elements not already in the prefix.
\begin{align}
\sigma_k &= \mathop{\arg\max}_{ h \in \{0, \dots, T \} \setminus \{ \sigma_0, \dots, \sigma_{k-1} \} } \frac{ \Exp[b(h, \sigma_{k-1})]}{1 - \gamma(h, \sigma_{k-1})} \label{eqn:dp_k}
\end{align}

\subsection{Analyzing the optimal policy} \label{sub:basic:analyzing}
While the above dynamic program produces a priority ordering $\sigma$ that defines an optimal policy, this in no way implies that the resulting optimal policy queries nodes in order of recency. Indeed, in~\cref{section:univariate}, we explore other regimes in which optimal policies exhibit very different structure. The main challenge in the following proof is to show that $\sigma = (0, 1, \dots, T)$, which implies that the optimal policy queries nodes in order of recency.

\paragraph{An observation about periods.} The following proof, in addition to many of the proofs in later sections, involves analyzing periods. For the purpose of analysis, it will be helpful to separate the immediate benefit returned by querying a single root from the benefit returned through recursively querying its descendants. Specifically, for a root $v$ with recency $h(v) = h$, we can think of an $(h, \sigma_j)$-period as the concatenation of two sub-periods, where $v$ is queried in the first sub-period and descendants of $v$ are queried in the second sub-period. Thus the first sub-period is equivalent to an $(h, \emptyset)$-period. If $v$ is infected, then at the start of the second sub-period the children of $v$ make up the frontier. Since $h(v) = h$, the children of $v$ have recencies defined by a random multiset $Z(h) = (Z_0, Z_1, \dots, Z_{h-1}) \sim D^h$, where $Z_j$ indicates the number of children of recency $j$. To describe the second sub-period, we first need to define a generalization of a period, called an epoch.

For any state $S = (X_0, X_1, \dots, X_T)$, define an $(S, \sigma_j)$-epoch as follows. At the start of step $t = 0$, the state $S$ defines the recencies of nodes present in the frontier. On each step $t \geq 0$ nodes are queried according to the ordering defined by the prefix $\sigma_0, \sigma_1, \dots, \sigma_j$ until all nodes $v'$ with recency $h(v') \preceq \sigma_j$ are exhausted. Let $b(S, \sigma_j)$ be the total discounted benefit accumulated over the period, let $\tau(S, \sigma_j)$ be the duration, and call $\gamma(S, \sigma_j) = \Exp[e^{-\beta \tau(S, \sigma_j)}]$ the expected pre-multiplier. 

Therefore an $(h, \sigma_j)$-period can be thought of as an $(h, \emptyset)$-period, which in the event that the root is infected, is followed by a $(Z(h), \sigma_j)$-epoch for $Z(h) \sim D^h$. If the root is infected, then a benefit of $e^{-\beta h}$ plus the benefit accumulated during the following epoch is returned. If the root is not infected then no benefit is returned. Therefore, the expected benefit accumulated over the $(h, \sigma_j)$-period is
\begin{align}
\Exp[b(h, \sigma_j)] &= p_T \left( e^{-\beta h} + e^{-\beta} \mathop{\Exp}_{Z(h) \sim D^h} \left[ b(Z(h), \sigma_j) \right] \right) \label{eqn:benefit}
\end{align}
Likewise $\gamma(h, \sigma_j)$ has a similar decomposition. If the root is infected, then $\tau(h, \sigma_j) = 1 + \tau(Z(h), \sigma_j)$ where $Z(h) \sim D^h$. If the root is not infected, then $\tau(h, \sigma_j) = 1$. Therefore
\begin{align}
\gamma(h, \sigma_j) &= \Exp \left[ e^{-\beta \tau(h, \sigma_j)} \right] \nonumber \\
&= p_T \mathop{\Exp}_{Z(h) \sim D^h} \left[ e^{-\beta (1 + \tau(Z(h), \sigma_j))} \right] + (1 - p_T) e^{-\beta} \nonumber \\
&= e^{-\beta} \left( p_T \mathop{\Exp}_{Z(h) \sim D^h} \left[ e^{-\beta \tau(Z(h), \sigma_j)} \right] + 1 - p_T \right) \label{eqn:premultiply}
\end{align}
For both decompositions the expectation is over the randomness of a node's descendants and their infection statuses.

\paragraph{Restating the optimal policy.} The probability of infection is $p_T$, and the benefit of querying an infected node of recency $h$ is $e^{-\beta h}$, so $\Exp[b(h, \emptyset)] = p_T e^{-\beta h}$. Applying this definition to~\cref{eqn:dp_0},
\begin{align}
\sigma_0 &= \mathop{\arg\max}_{ h \in \{ 0, 1, \dots, T \} } p_T e^{-\beta h} \label{eqn:basic_dp0}
\end{align}
In any round $k > 0$, the prefix $\sigma_0, \sigma_1, \dots, \sigma_{k - 1}$ is fixed, and $\sigma_k$ is selected from the elements not already in the prefix. Applying~\cref{eqn:benefit,eqn:premultiply} to~\cref{eqn:dp_k},
\begin{align}
\sigma_k &= \mathop{\arg\max}_{ h \in \{0, \dots, T \} \setminus \{ \sigma_0, \dots, \sigma_{k-1} \} } \frac{ p_T \left( e^{-\beta h} + e^{-\beta} \mathop{\Exp}_{Z(h) \sim D^h} \left[ b(Z(h), \sigma_{k - 1}) \right] \right) }{1 - e^{-\beta} \left( p_T \mathop{\Exp}_{Z(h) \sim D^h} \left[ e^{-\beta \tau(Z(h), \sigma_j)} \right] + 1 - p_T \right) } \label{eqn:basic_dpk}
\end{align}
With these decompositions we are now equipped to prove~\cref{thm:basic_main}. 
\setcounter{thm}{0}
\basicmain
\begin{proof}
Fix $T \in \NN$, $p_T \in (0, 1]$, and $\beta > 0$. Let $\sigma$ be the optimal priority ordering constructed via the dynamic program in~\cref{sub:basic:defining}. It suffices to show that for all $j \in \{ 0, 1, \dots, T \}$, $\sigma_j = j$. 

Proof by induction on $j$. By~\cref{eqn:basic_dp0},  
\begin{align*}
\sigma_0 &= \mathop{\arg\max}_{ h \in \{ 0, 1, \dots, T \} } p_T e^{-\beta h} = 0
\end{align*}
Fix $0 < k < T$. Assume that for all $0 \leq j \leq k-1$, $\sigma_j = j$, resulting in the prefix is $(0, 1, \dots, k-1)$. Fix $h \geq k$. For $Z(h) = (Z_0, Z_1, \dots, Z_{h-1}) \sim D^h$, consider a $(Z(h), \sigma_{k-1})$-epoch. The epoch begins in state $Z(h)$, and since this particular prefix is in order of recency, nodes are then queried recursively in order of recency until all nodes $v'$ of recency $h(v') \leq k-1$ are exhausted. Thus a $(Z(h), \sigma_{k - 1})$-epoch only queries nodes with recencies in $\{ 0, 1, \dots, k - 1 \}$. Let $\Pi_k(Z(h)) = (Z_0, Z_1, \dots, Z_{k-1})$ be the projection of $Z(h)$ onto the first $k$ coordinates, and note that $\Pi_k(Z(h)) \sim D^k$. Thus the benefit and duration of a $(Z(h), \sigma_{k - 1})$-epoch are identically distributed as the benefit and duration of a $(\Pi_k(Z(h)), \sigma_{k - 1})$-epoch. Call the total expected benefit $b_k$ and observe that it does not depend on $h$.
\begin{align}
\mathop{\Exp}_{Z(h) \sim D^h} \left[ b(Z(h), \sigma_{k-1}) \right] = \mathop{\Exp}_{Z(h) \sim D^h} \left[ b(\Pi_k(Z(h)), \sigma_{k-1}) \right] = \mathop{\Exp}_{Z(k) \sim D^k} \left[ b(Z(k), \sigma_{k-1}) \right] = b_k \label{eqn:bk}
\end{align}
Since the durations are also identically distributed, the expected pre-multipliers are equal. Call this pre-multiplier $\gamma_k$ and observe that it has does not depend on $h$.
\begin{align}
\mathop{\Exp}_{Z(h) \sim D^h} \left[ \gamma(Z(h), \sigma_{k-1}) \right] = \mathop{\Exp}_{Z(h) \sim D^h} \left[ \gamma(\Pi_k(Z(h)), \sigma_{k-1}) \right] = \mathop{\Exp}_{Z(k) \sim D^k} \left[ \gamma(Z(k), \sigma_{k-1}) \right] = \gamma_k \label{eqn:gk}
\end{align}
Combining~\cref{eqn:benefit,eqn:bk}, the expected benefit of an $(h, \sigma_{k-1})$-period is
\begin{align*}
\Exp[b(h, \sigma_{k - 1})] &= p_T \left( e^{-\beta h} + e^{-\beta} \mathop{\Exp}_{Z(h) \sim D^h} \left[ b(Z(h), \sigma_{k - 1}) \right] \right) \\
&= p_T \left( e^{-\beta h} + e^{-\beta} b_k \right) \\
\end{align*}
Likewise, combining~\cref{eqn:premultiply,eqn:gk}, the expected pre-multiplier is
\begin{align*}
\gamma(h, \sigma_{k - 1}) &= e^{-\beta} \left( p_T \mathop{\Exp}_{Z(h) \sim D^h} \left[ e^{-\beta \tau(Z(h), \sigma_{k - 1})} \right] + 1 - p_T \right) \\
&= e^{-\beta} \left( p_T \gamma_k + 1 - p_T \right)
\end{align*}
By~\cref{eqn:basic_dpk}, given the prefix $(0, 1, \dots, k-1)$, $\sigma_k$ is selected by comparing $(h, \sigma_{k-1})$-periods for $h \in \{ k, \dots, T \}$. Therefore for all $h$ in consideration, the above identities hold. Then
\begin{align*}
\sigma_k &= \mathop{\arg\max}_{ h \in \{k, \dots, T \} } \frac{ \Exp[b(h, \sigma_{k-1})]}{1 - \gamma(h, \sigma_{k-1})} \\
&= \mathop{\arg\max}_{ h \in \{ k, \dots, T \} } \frac{ p_T \left( e^{-\beta h} + e^{-\beta} b_k \right) }{1 - e^{-\beta} \left( p_T \gamma_k + 1 - p_T \right) } \\
&= k
\end{align*}
To understand the final equality, note that $e^{-\beta h}$ decreases with $h$, and all other terms are constant with respect to $h$. Thus $k$, the smallest value of $h$ in consideration, attains the maximum. Since $\sigma_k = k$, for all $j \in \{ 0, 1, \dots, T \}$, $\sigma_j = j$. 
\end{proof}
Thus, despite the fact that the optimal policy in the basic setting is simply querying nodes in order of recency, the proof of optimality requires balancing trade-offs between recency and benefit.

%% file: sections/univariate.tex
\section{Univariate Model} \label{section:univariate}
While in the basic model querying nodes in order of recency is optimal, this is not always the case in the univariate model. In the univariate model the probability of infection decays absolutely with time according to an exponential functional form.

Fixing $p_T \in (0, 1]$ and $\alpha \geq 0$, for any $h \in \{ 0, 1, \dots, T \}$ the probability of infection is $p(h) = p_T e^{-\alpha (T - h)}$. All other parameters of the univariate model are identical to the basic model. Observe that if $\alpha = 0$ the probability of infection is constant and the univariate model instantiates the basic model. Additionally, $p_T$ represents the probability that a node of recency $T$ is infected.

\paragraph{Examining trade-offs and thresholds.} As in the basic model, there is a trade-off between querying a less recent node, which provides an opportunity to significantly expand the frontier, and querying a more recent node which, if infected, returns a larger immediate benefit. However, since the probability of infection decays with time, a more recent node has a lower probability of infection. As a result, querying a more recent node does not necessarily return a higher \textit{expected} immediate benefit. Thus choosing a node to query involves a trade-off between a node's recency and probability of infection. 

The calculation of this trade-off changes with $\alpha$. If $\alpha = 0$ the probability of infection is constant, so as in the basic model, querying nodes in order of recency is optimal. As $\alpha$ increases it eventually hits a threshold beyond which this policy is no longer optimal, at which point it is natural to wonder whether any structure remains. In fact, we show that the there is still structure to the optimal policy: the policy always queries either the most recent or least recent node available. We say that such a policy is defined by an \textit{interleaved} priority ordering.
\interleavingdef
Observe that many different priority orderings satisfy the interleaving property. For example, an ordering that prioritizes nodes by recency is interleaved, as is an ordering that prioritizes nodes by reverse recency. Once $\alpha = \beta$, each node returns the same expected benefit, and we show that in this case any priority ordering is optimal. Once $\alpha > \beta$, a less recent node both returns a higher expected immediate benefit and has more children in expectation, and we show that in this regime querying nodes in order of reverse recency is optimal. Taken together, these results imply that the interleaving property holds for all instances of the univariate model.

\univmain
\begin{proof}
The main challenge is proving that the interleaving property holds for $0 < \alpha < \beta$, which we show in~\cref{thm:univ_interleaved}. If $\alpha = 0$, querying nodes in order of recency is optimal, by~\cref{thm:basic_main}. If $\alpha \geq \beta$, then querying nodes in order of reverse recency is optimal, as shown in~\cref{thm:univ_equal,thm:univ_reverse}. In general, the proof techniques involve analyzing periods and follow a similar structure as the proof in~\cref{thm:basic_main}.
\end{proof}

\subsection{Defining the optimal policy} \label{sub:univ:dp}
The optimal policy is constructed via the same dynamic program as defined in~\cref{eqn:basic_dp0,eqn:basic_dpk}, except the probability of infection is now determined by the function $p(h)$.
\begin{align}
\sigma_0 &= \mathop{\arg\max}_{ h \in \{ 0, 1, \dots, T \} } p(h) e^{-\beta h} \label{eqn:univ:dp0}
\end{align}
In any round $k > 0$, the prefix $\sigma_0, \sigma_1, \dots, \sigma_{k - 1}$ is fixed, and $\sigma_k$ is selected from the elements not already in the prefix.
\begin{align}
\sigma_k &= \mathop{\arg\max}_{ h \in \{0, \dots, T \} \setminus \{ \sigma_0, \dots, \sigma_{k-1} \} } \frac{ p(h) \left( e^{-\beta h} + e^{-\beta} \mathop{\Exp}_{Z(h) \sim D^h} \left[ b(Z(h), \sigma_{k - 1}) \right] \right) }{1 - e^{-\beta} \left( p(h) \mathop{\Exp}_{Z(h) \sim D^h} \left[ e^{-\beta \tau(Z(h), \sigma_{k - 1})} \right] + 1 - p(h) \right) } \label{eqn:univ:dpk}
\end{align}

\subsection{Analyzing the optimal policy}
We analyze the above construction via similar strategies as in the proof of~\cref{thm:basic_main}.
\begin{restatable}{theorem}{univinterleaved} \label{thm:univ_interleaved}
In the univariate model, for $0 < \alpha < \beta$, there is an optimal policy defined by an interleaved priority ordering. 
\end{restatable}
\begin{proof}
Fix $T \in \NN$, $p_T \in (0, 1]$, $\beta > 0$, and $0 < \alpha < \beta$. Let $\sigma$ be the optimal priority ordering constructed via the dynamic program in~\cref{sub:univ:dp}. It suffices to show that for all $0 \leq j \leq T$, $\sigma_j \in \{ \max_{i \geq j} \sigma_i, \min_{i \geq j} \sigma_i \}$.

Proof by induction. By definition,
\begin{align*}
\sigma_0 &= \arg\max_{ h \in \{ 0, 1, \dots, T \} } p(h) e^{-\beta h} \\
&= \arg\max_{ h \in \{ 0, 1, \dots, T \} } p_T e^{-\alpha T} e^{(\alpha - \beta) h} \\
&= 0
\end{align*}
Here $0$ attains the maximum, since $\alpha < \beta$.

Fix $0 < k < T$. Assume that for all $0 \leq j \leq k-1$, $\sigma_j \in \{ \max_{i \geq j} \sigma_i, \min_{i \geq j} \sigma_i \}$. Then there exist elements $l, m \in \{ 0, 1, \dots, T \}$ with $0 < l \leq m$ such that $\{ l, \dots, m \} = \{ 0, 1, \dots, T \} \setminus \{ \sigma_0, \sigma_1, \dots, \sigma_{k - 1} \}$. Fix $h \in \{ l, \dots, m \}$ and let $v$ be a node with recency $h(v) = h$. Then the children of $v$ are described by a multiset $Z(h) = (Z_0, Z_1, \dots, Z_{h-1}) \sim D^h$, where $Z_j$ indicates the number of children of recency $j$. Consider a $(Z(h), \sigma_{k - 1})$-epoch. Since a child is more recent than its parent, any descendant $u$ of $v$ has recency $h(u) \leq h - 1$. As a result, a $(Z(h), \sigma_{k - 1})$-epoch involves only nodes with recencies in $\{ 0, 1, \dots, h - 1 \}$. Additionally, the prefix $\sigma_0, \sigma_1, \dots, \sigma_{k - 1}$ dictates that any node $u$ queried during the epoch has recency $h(u) \preceq \sigma_{k - 1}$. Combining these two restrictions, if a node $u$ is queried during a $(Z(h), \sigma_{k - 1})$-epoch then
\begin{align*}
h(u) \in \{ \sigma_0, \sigma_1, \dots, \sigma_{k - 1} \} \cap \{ 0, 1, \dots, h - 1 \} = \{ 0, 1, \dots, l - 1 \}
\end{align*}
Let $\Pi_l(Z(h)) = (Z_0, Z_1, \dots, Z_{l - 1})$ be the projection of $Z(h)$ onto the first $l$ coordinates, and note that $\Pi_l(Z(h)) \sim D^l$. Following the same reasoning as in~\cref{thm:basic_main}, the benefit and duration of a $(Z(h), \sigma_{k - 1})$-epoch are identically distributed as the benefit and duration of a $(\Pi_l(Z(h)), \sigma_{k - 1})$-epoch. Call the total expected benefit $b_l$ and observe that it does not depend on $h$.
\begin{align}
\mathop{\Exp}_{Z(h) \sim D^h} \left[ b(Z(h), \sigma_{k-1}) \right] = \mathop{\Exp}_{Z(h) \sim D^h} \left[ b(\Pi_l(Z(h)), \sigma_{k-1}) \right] = \mathop{\Exp}_{Z(l) \sim D^l} \left[ b(Z(l), \sigma_{k-1}) \right] = b_l \label{eqn:univ:benefit}
\end{align}
Likewise, since the duration is identically distributed, the expected pre-multipliers are equal. Call this pre-multiplier $\gamma_l$, and observe that it does not depend on $h$.
\begin{align}
\mathop{\Exp}_{Z(h) \sim D^h} \left[ \gamma(Z(h), \sigma_{k-1}) \right] = \mathop{\Exp}_{Z(h) \sim D^h} \left[ \gamma(\Pi_l(Z(h)), \sigma_{k-1}) \right] = \mathop{\Exp}_{Z(l) \sim D^l} \left[ \gamma(Z(l), \sigma_{k-1}) \right] = \gamma_l \label{eqn:univ:premultiply}
\end{align}
Given the prefix $\sigma_0, \sigma_1, \dots, \sigma_{k-1}$, the element $\sigma_k$ is selected by comparing $(h, \sigma_{k-1})$-periods for $h$ not in the prefix, that is, for $h \in \{ l, \dots, m \}$. Therefore for all $h$ in consideration~\cref{eqn:univ:benefit,eqn:univ:premultiply} hold. Applying these identities and the definition of $p(h)$ to~\cref{eqn:univ:dpk},
\begin{align*}
\sigma_k &= \mathop{\arg\max}_{ h \in \{ l, \dots, m \} } \frac{ p(h) \left( e^{-\beta h} + e^{-\beta} b_l \right) }{1 - e^{-\beta} \left( p(h) \gamma_l + 1 - p(h) \right) } \\
&= \mathop{\arg\max}_{ h \in \{ l, \dots, m \} } \frac{ p_T e^{-\alpha (T - h)} \left( e^{-\beta h} + e^{-\beta} b_l \right) }{1 - e^{-\beta} + p_T e^{-\alpha (T - h) - \beta} (1 - \gamma_l)}
\end{align*}
Define the function $I_l: \RR \rightarrow \RR$, where
\begin{align}
I_l(x) & = \frac{ p_T e^{-\alpha (T - x)} \left( e^{-\beta x} + e^{-\beta} b_l \right) }{1 - e^{-\beta} + p_T e^{-\alpha (T - x) - \beta} (1 - \gamma_l)} \label{eqn:univ:il}
\end{align}
Then 
\begin{align}
\sigma_k &= \mathop{\arg \max}_{h \in \{ l, \dots, m \}} I_l(h) \label{eqn:univ:interval_ub}
\end{align}
By~\cref{lem:univ_increasing} $I_l$ is convex. This implies that the maximum value of $I_l$ on any closed interval is attained at an endpoint. Therefore
\begin{align}
\sigma_k = \mathop{\arg\max}_{x \in \{ l , \dots, m \} } I_l(x) = \mathop{\arg\max}_{x \in [l, m]} I_l(x) = \mathop{\arg\max}_{x \in \{ l, m \}} I_l(x)
\end{align}
As a result, $\sigma_k \in \{ l, m \}$, so $\sigma_k \in \{ \min_{i \geq k} \sigma_i, \max_{i \geq k} \sigma_i \}$. Therefore $\sigma$ is interleaved.
\end{proof}
The following lemma completes the above proof.
\begin{lemma} \label{lem:univ_increasing}
$I_l$ is convex.
\end{lemma}
\begin{proof}
Fix $T \in \NN$, $p_T \in (0, 1]$, $\beta > 0$, and $0 < \alpha < \beta$. Fix $b_l, \gamma_l \geq 0$. It suffices to show that $I_l'$ is non-decreasing.

Set $C_0 = p_T e^{-\alpha T}$, $C_1 = e^{-\beta} b_l$, $C_2 = 1 - e^{-\beta}$, and $C_3 = p_T e^{-\alpha T - \beta} (1 - \gamma_l)$. Given these parameters, we can rewrite $I_l$ as
\begin{align*}
I_l(x) &= \frac{ C_0 e^{\alpha x} \left( e^{-\beta x} + C_1 \right) }{C_2 + C_3 e^{\alpha x} }
\end{align*}
Then
\begin{align*}
I_l'(x) &=  \frac{C_0  e^{\alpha x}}{(C_2 + C_3 e^{\alpha x})^2} \cdot \left[ C_2 (\alpha - \beta) e^{-\beta x} - C_3 \beta e^{(\alpha - \beta) x} + \alpha C_1 C_2 \right]
\end{align*}
Observe that $C_0, C_2, C_3 > 0$ and $C_1 \geq 0$. Examining $I_l'(x)$, note that the third term is positive and has no dependence on $x$. Since $0 < \alpha < \beta$, the first two terms are negative and both decrease in absolute value as $x$ increases. Thus $I_l'(x)$ is non-decreasing, and therefore $I_l$ is convex.
\end{proof}
In the special case when $\alpha = \beta$, the order in which nodes are queried has no effect on the total expected benefit, and any policy is optimal.
\begin{restatable}{theorem}{univequal} \label{thm:univ_equal}
If $\alpha = \beta$, any policy is optimal.
\end{restatable}
\begin{proof}
Fix $T \in \NN$, $p_T \in (0, 1]$, and $\alpha = \beta > 0$. It suffices to show that, starting from an arbitrary state $S$, any arbitrary policies $P_1$ and $P_2$ achieve the same total expected benefit. 

Let $b(S, P_1)$ and $b(S, P_2)$ be random variables indicating the total discounted benefit accumulated by running $P_1$ and $P_2$, respectively, from the initial state $S$. Since $\alpha = \beta$, at any step $t$ the expected discounted benefit of querying any node in the frontier is 
\begin{align*}
p(h) e^{-\beta (t + h)} = p_T e^{-\alpha T + (\alpha - \beta) h - \beta t} = p_T e^{-\alpha T - \beta t}
\end{align*}
As a result, a policy that queries a node at step $t$ receives an expected discounted benefit of $p_T e^{-\alpha T - \beta t}$. Thus the total expected benefit a policy achieves is defined by the first step in which the frontier is empty, since from then on no nodes are queried. Let $\tau_1$ and $\tau_2$ be random variables indicating the first step in which the frontier is empty for $P_1$ and $P_2$, respectively, starting from state $S$. Then 
\begin{align*}
\Exp \left[ b(S, P_1) \right] &= \sum_{t = 0}^{\tau_1 - 1} p_T e^{-\alpha T - \beta t} \\
\Exp \left[ b(S, P_2) \right] &= \sum_{t = 0}^{\tau_2 - 1} p_T e^{-\alpha T - \beta t}
\end{align*}
The order in which nodes are queried does not affect the first step in which the frontier is empty, so $\tau_1 \sim \tau_2$. Therefore $\Exp[b(S, P_1)] = \Exp[b(S, P_2)]$.
\end{proof}
Once $\alpha > \beta$, a less recent node has both a higher expected immediate benefit and more children in expectation, so querying nodes in order of reverse recency is optimal.
\begin{restatable}{theorem}{univreverse} \label{thm:univ_reverse}
If $\alpha > \beta$, it is optimal to query nodes in order of reverse recency.
\end{restatable}
\begin{proof}
Let $\sigma$ be the optimal priority ordering constructed via the dynamic program in~\cref{sub:univ:dp}. It suffices to show that for all $ 0 \leq j \leq T$, $\sigma_j = T - j$. 

Proof by induction. By definition,
\begin{align*}
\sigma_0 &= \mathop{\arg\max}_{ h \in \{ 0, 1, \dots, T \} } p(h) e^{-\beta h} \\
&= \mathop{\arg\max}_{ h \in \{ 0, 1, \dots, T \} } p_T e^{-\alpha T + (\alpha - \beta) h} \\
&= T
\end{align*}
Here $T$ attains the maximum, since $\alpha > \beta$.

Fix $0 < k < T$. Assume that that for all $0 \leq j \leq k - 1$, $\sigma_j = T - j$. As a result,
\begin{align*}
(\sigma_0, \sigma_1, \dots, \sigma_{k - 1}) = (T, T - 1, \dots, T - k + 1)
\end{align*}
Given this prefix, $\sigma_k$ is selected by comparing $(h, \sigma_{k - 1})$-periods for $h$ not in the prefix, that is, for $h \in \{ 0, 1, \dots, T - k \}$. Fix $h \in \{ 0, 1, \dots, T - k \}$, and let $v$ be a node with recency $h(v) = h$. Recall that an $(h, \sigma_{k - 1})$-period consists of an $(h, \emptyset)$-period followed by a $(Z(h), \sigma_{k-1})$-epoch, for $Z(h) = (Z_0, Z_1, \dots, Z_{h - 1}) \sim D^h$. Since a node is less recent than its descendants, any descendant $v'$ of $v$ has recency $h(v') < h \leq T - k$, so the epoch only involves nodes with recencies in $\{ 0, 1, \dots, T - k \}$. However, the prefix dictates that any node $u$ queried during the epoch has recency $h(u) > T - k$. Thus no node is queried during the epoch, so for this particular prefix an $(h, \sigma_{k - 1})$-period is equivalent to an $(h, \emptyset)$-period.

Therefore
\begin{align*}
\sigma_k &= \mathop{\arg \max}_{h \in \{ 0, 1, \dots, T - k \} } \frac{\Exp[b(h, \sigma_{k-1})]}{1 - \gamma(h, \sigma_{k-1})} \\
&= \mathop{\arg \max}_{h \in \{ 0, 1, \dots, T - k \} } \frac{\Exp[b(h, \emptyset)]}{1 - \gamma(h, \emptyset)} \\
&= \mathop{\arg \max}_{h \in \{ 0, 1, \dots, T - k \} } \frac{p_T e^{-\alpha T + (\alpha - \beta) h}}{1 - e^{-\beta} } \\
&= T - k
\end{align*}
Here $T - k$ attains the maximum, because $\alpha > \beta$. Thus for all $0 \leq j \leq T$, $\sigma_j = T - j$, so the optimal policy queries nodes in order of reverse recency.
\end{proof}

%% file: sections/bivariate.tex
\section{Bivariate Model} \label{section:bivariate}
In the bivariate model the probability that a node is infected decays relative to the incubation period of its parent. Modeling this requires keeping track of two parameters, a node's recency and the recency of its parent. We say that a node of recency $h$ with a parent of recency $h'$ has \textit{span} $\Delta = h' - h$. The probability that a node of span $\Delta$ is infected is $p(\Delta) = p_T e^{-\alpha \Delta}$. Aside from the probability of infection, the bivariate model is identical to the basic model, so a node's recency $h$ determines the distribution on its children and the benefit returned if it is found to be infected. A node of recency $h$ and span $\Delta$ is defined by its \textit{type} $(h, \Delta)$, and policies in the bivariate model are on the set of types $\{ 0, 1, \dots, T \}^2$. Thus the bivariate model demonstrates an analysis of policies that take into account multiple parameters. When $\alpha = 0$, the bivariate model instantiates the basic model.

Just as in the basic and univariate models we assume that nodes of the same recency are indistinguishable until they are queried, in the bivariate model we assume that nodes of the same type are indistinguishable until they are queried. Therefore we redefine the state $S_t$ to be the multiset of types present in the frontier at time $t$. If the tracer queries an infected node, the query reveals the type of each contact. Recall that a node $v$ of recency $h(v) = h$ has a multiset of children $Z(h) = (Z_0, Z_1, \dots, Z_{h - 1}) \sim D^h$, where $Z_j$ indicates the number of children of recency $j$. Since $h(v) = h$, a child with recency $j$ has span $i = h - j$. Therefore the children of $v$ are defined by the multiset of types $Y(h) = (Y_{0, h}, Y_{1, h-1}, \dots, Y_{h-1, 1}) \sim D^h$, where $Y_{j,i}$ indicates the number of children with recency $j$ and span $i$. Observe that $i + j = h$, which implies that nodes of different recencies have no children of the same type.

Our main result in this section is a monotonicity property that shows that it is optimal to query nodes of the same recency in order of increasing span.
\bivspan
To prove this result, we analyze an optimal priority ordering using many of the same techniques developed in~\cref{section:basic,section:univariate}. The proof leverages the fact that nodes of the same recency have the same distribution on descendants in order to optimize over span. 

As a complement to the main result, we also examine monotonicity along the dimension of recency. Since nodes of different recencies have no children of the same type, comparing periods becomes tricky, and the inductive approaches used in~\cref{section:basic,section:univariate} do not seem to apply here. Instead we construct the optimal ordering of types step-by-step. To make this approach tractable, our result is restricted to settings where the contact distribution $D$ is a Bernoulli distribution and where all nodes have recencies in $\{ 0, 1, 2 \}$. 
\bivrecency
While the other restrictions are due to our particular approach, a lower bound on $\beta$ is in fact necessary; there are settings of $\beta$ for which querying nodes of the same span in order of recency is not optimal. Full proofs of the above theorems are in~\cref{app:biv}.

%% file: sections/general.tex
\section{General Model} \label{section:general}
The general model provides a broad framework that allows for a variety of factors to affect how individuals interact and how the infection spreads. In the general model each node is assigned an arbitrary type, which is associated with an arbitrary probability of infection and an arbitrary distribution on descendants. An individual’s type could be thought of as representing information like their profession or role within a community and the context in which they are exposed. Thus the general model has the flexibility to describe the categorizations of contacts we see in practice, where the priority assigned to a contact depends on multiple different factors. The general model also encompasses the basic, univariate, and bivariate models.

We show in~\cref{thm:general_reduction} that the general model reduces to the branching bandit model. This has two main implications. First, it formally defines a connection between contact tracing and the brancing bandit model. Second, it implies that the optimal policy for any instance of the general model can be found by analyzing the optimal policy for a corresponding instance of the branching bandit model, as constructed by the dynamic program in~\cref{sub:basic:weiss}. We show in~\cref{thm:general} that, as a result, any instance of the general model has an optimal policy which is an efficiently computable index policy on the set of types with an efficient construction. Since the basic, univariate, and bivariate models are all instances of the general model, this result applies to these settings as well, and it defines the construction of optimal policies in each of the three models.

\subsection{Model}
The general model follows the same structure as the model from~\cref{section:model}, however the changes to the contact and infection processes warrant a second overview. 

\paragraph{Phase 1} Phase 1 spans steps $t = -T$ through $t = 0$. During this phase each individual meets new contacts and infected individuals probabilistically infect each new contact they meet.

Each individual is associated with a known role $w \in \{ 1, 2, \dots, W \}$ and a hidden binary infection status $d \in \{ 0, 1 \}$. Each individual also has a known recency $h \in \{ 0, 1, \dots, T \}$. An individual with recency $h$ and role $w$ belongs to category $(h, w) \in C$ where $C = \{ 1, 2, \dots, T \} \times \{ 1, 2, \dots, W \}$. After an individual is exposed, on each step $t$ thereafter they meet new contacts defined by a multiset of categories drawn from $D_{c,d,t}$, a distribution on all multisets of elements in $[C]$. If the individual is infected, they infect each new contact they meet independently according to the function $p: [C] \times [C] \rightarrow [0, 1]$. The probability that an infected individual from category $c'$ infects a contact in category $c$ is $p(c', c)$. As in the model from~\cref{section:model}, we are agnostic to the origin of the index cases, and we model each index case as the child of an infected super-root $\bot$. An index case with category $c$ is infected with probability $p(c(\bot), c)$.

\paragraph{Step $t = 0$.} On step $t = 0$ the contact and infection processes halt and from then on no new infections occur.

To understand the system on step $t = 0$, consider an individual $v$ from category $c(v) = c$. An individual in category $c$ has recency $h(c)$, so $v$ was exposed in step $-h(c)$. If $v$'s exposure results in infection, for each step $-h(c) + 1 \leq t \leq 0$ in the remainder of Phase 1, $v$ exposes $(Z_1(t), \dots, Z_C(t)) \sim D_{c, 1, t}$ individuals, where $Z_j(t)$ indicates the number of individuals with category $j$ exposed on step $t$. Then by the end of step $t = 0$, throughout the course of the first phase $v$ has exposed a multiset of contacts $ Z(c) = (Z_1, \dots, Z_C)$, where $Z_j = \sum_{t = -h(c)+1}^0 Z_j(t)$. Let $D_{c}$ be the distribution on $Z(c)$.

As in~\cref{section:model}, we can view $v$ as a node in a tree of potential exposures, where the children of $v$ are all the contacts $v$ met after being exposed. Let $u$ be the node which exposed $v$. Then with probability $p((c(u), c(v))$ node $v$ is infected and has children with categories $(Z_1, \dots, Z_C) \sim D_c$.

\paragraph{Phase 2} Phase 2 begins on step $t = 0$ and continues indefinitely. During Phase 2 an individual's infection status is fixed but hidden. A contact tracer seeks to identify infected individuals as efficiently as possible.

As in~\cref{section:model}, contact tracing begins on step $t = 0$ when a set of index cases are identified. Initally the index cases are the only nodes available to query, and the tracer observes the category of each index case. From then on, at each step $t \geq 0$ the tracer selects one node from the frontier to query. Querying a node reveals its infection status, and if they are infected, its contacts are included in the frontier. The tracer may only query a node that is an index case or the contact of an infected node already queried. Querying an infected node $v$ with category $c(v) = c$ at step $t$ returns the benefit $b(c, t) = b(c) e^{-\beta t}$, for some arbitrary contant $b(c) \in [0, 1]$. 

Now consider the information available to the tracer at each step. Any node $v$ in the frontier has an infected parent $u$, where $u$ is either the super-root or another node in the tree that was already queried. Thus the tracer knows both $v$'s category $c(v)$ and $u$'s category $c(u)$. As described above, these two parameters define the probability that $v$ is infected, and if $v$ is infected, the benefit of querying $v$ and the distribution on children added to the frontier. That is, $c(u)$ and $c(v)$ fully define the distribution on outcomes that result from querying $v$. We say that $v$ has \textit{type} $(c(u), c(v)) \in [C]^2$. As in~\cref{section:bivariate}, we assume two nodes of the same type are indistinguishable until they are queried. Then the state $S_t$ is the multiset of types present in the frontier at time $t$, and at each step the tracer selects a node to query based on its type. 

\paragraph{Defining the objective.} On each step $t \geq 0$, the tracer selects a node $v_t$ from the frontier to query where $v_t$ has category $c(v_t)$. Since $v_t$ is in the frontier, it was exposed to the infection. Let $\mathbbm{1}(v_t)$ indicate whether $v_t$ is infected $(1)$ or uninfected $(0)$. If $v_t$ is infected, then benefit $b(c(v_t)) e^{-\beta t}$ is returned. Thus the total benefit the tracer accumulates over the course of Phase 2 is
\begin{align*}
\sum_{t \geq 0} \mathbbm{1}(v_t) \cdot b(c(v_t)) e^{-\beta t}.
\end{align*}
As in~\cref{section:model}, the objective is to develop a policy for querying nodes that maximizes the total expected benefit, where the expectation is taken over all realizations $\{ \mathbbm{1}(v_0), \mathbbm{1}(v_1), \mathbbm{1}(v_2), \dots \}$.

\paragraph{Defining the general model from types.} The model can be equivalently described in terms of types instead of categories, which helps simplify the proof that follows. Enumerate the set of all types $[C]^2$ as $\{ 1, 2, \dots, C^2 \}$. Let $v$ be a node of type $(c', c) \in [C]^2$, where $(c', c)$ is indexed as type $j \in \{1, 2, \dots, C^2 \}$ in the enumeration. We now define the probability of infection, benefit, and distribution on children associated with type $j$. If $v$ is exposed, the probability that $v$ is infected is $\overline{p}(j) = p(c', c)$. If $v$ is infected, the benefit of querying $v$ is $\overline{b}(j) = b(c)$ and $v$'s children have categories $(Z_1, \dots, Z_{C}) \sim D_c$, where $Z_k$ indicates the number of children from category $k$. Since $c(v) = c$, any child of $v$ from category $k$ has type $(c, k)$, so $Z_k$ equivalently indicates the number of children of type $l = (c, k)$. Let $\overline{D}_j$ be the equivalent distribution on multisets of types in $\{ 1, 2, \dots, C^2 \}$ so that for $(Y_1, \dots, Y_{C^2}) \sim \overline{D}_j$, $Y_l$ indicates the number of children of type $l$. Then the state $S_t = (X_1, \dots, X_{C^2})$ is the multiset of types present in the frontier at time $t$.

\subsection{Optimal policies in the general model}
Our primary result shows that the optimal policy in the general model is an index policy on the set of types $\{ 1, 2, \dots, C^2 \}$. 
\generalpolicy
\begin{proof}
Proof via reduction to the branching bandit model. Recall that the branching bandit model involves arms belonging to classes $\{ 1, \dots, L \}$. When an arm of class $i$ is pulled it yields a reward $R(i)$ and is replaced by a set of new arms $N_{i1}, \dots, N_{iL}$, where each class $i$ has an arbitrary, known, joint distribution on the random variables $R(i)$ and $N_{i1}, \dots, N_{iL}$. Additionally, each class $i$ is also associated with a random variable $\mu(i)$, where pulling an arm from class $i$ occupies $\mu(i)$ steps, however for the purposes of this reduction we consider only a restricted model where each pull occupies exactly $1$ step. At step $t$ the system is defined by a vector $n(t) = (N_1, \dots, N_L)$, where $N_i$ is the total number of arms of class $i$ available, and a reward received at step $t$ is discounted by $e^{-\eta t}$ for a fixed parameter $\eta > 0$. As described in~\cref{sub:basic:weiss}, the optimal policy in the branching bandit model is an index policy on the set of classes $\{ 1, \dots, L \}$ with an efficient construction.

The plan for the reduction is to map each node to an arm. The idea is to map the benefit and new children returned by querying a node to the reward and new arms returned by pulling an arm. To do this, for each type in the contact tracing instance we construct a class in the branching bandit instance with an appropriate distribution on reward and new arms.~\Cref{thm:general_reduction} proves that, for a specific mapping from types to classes, the general model reduces to the branching bandit model. Informally, proving the reduction involves showing that querying a sequence of nodes maps to pulling a corresponding sequence of arms.

The claim follows from this reduction. Since types map to classes, at any step $t$ the state of the frontier $S_t$ is represented by a vector of arms $n(t)$. The optimal policy in the branching bandit model selects the next arm to pull from $n(t)$, which by the reduction dictates the next node to query from the frontier. Thus the optimal policy in the branching bandit model defines the optimal policy in the general model. In particular, since the former is an index policy on the set of classes, the latter is an index policy on the set of types. Therefore the optimal policy in the general model is an index policy on the set of types with an efficient construction.
\end{proof}
Before continuing with the reduction, we first describe what proving such a reduction requires. We start by defining the general model and the branching bandit model each as games where an agent chooses actions in order to transition between different states. (Here ``state'' is used in the general sense of an agent moving through different states in a game, not as a multiset of types.) Proving a reduction involves mapping a sequence of states and actions in the general model to a corresponding sequence of states and actions in the branching bandit model. 

In the general model, the tracer (the agent) chooses types (actions) and in the branching bandit model the agent chooses classes (actions). Specifically, in the general model the state at step $t$ is a pair $(S_t, b_t)$ where $b_t$ is the total benefit accumulated through the start of step $t$. The tracer chooses a type $j_t \in S_t$, and in response a random benefit $B(j_t) e^{-\beta t}$ is returned and a multiset of new children $(Y_1, \dots, Y_{C^2}) \sim D_{j_t}$ is added to the frontier. As a result, the tracer transitions to the state $(S_{t + 1}, b_{t + 1})$ where $S_{t + 1} = (S_t \setminus j_t) \cup (Y_1, \dots, Y_{C^2})$ and $b_{t + 1} = b_t + B(j_t) e^{-\beta t}$. In the branching bandit model, the state at step $t$ is a pair $(n(t), r_t)$ where $r_t$ is the total reward earned through the start of step $t$. The agent chooses a class $i_t \in n(t)$, and in response a random reward $R(i_t) e^{-\eta t}$ is received and a multiset of new arms $(N_{i_t 1}, \dots, N_{i_t L})$ are added to the system. As a result, the agent transitions to the state $(n(t + 1), r_{t + 1})$ where $n(t + 1) = (n(t) \setminus i_t) \cup (N_{i_t 1}, \dots, N_{i_t L})$ and $r_{t + 1} = r_t + R(i_t) e^{-\eta t}$.

Let $(S_0, b_0)$ be an arbitrary initial state, and let $(S_0, b_0, j_0), \dots, (S_{t'}, b_{t'}, j_{t'})$ be a sequence of states and actions. The general model reduces to the branching bandit model if there is a mapping of states and actions such that the next-state distribution is preserved. Specifically, we need a mapping from types to classes so that a corresponding sequence $(n(0), r_0, i_0), \dots, (n(t'), r_{t'}, i_{t'})$ exists which fulfills the following criteria for all $0 \leq t \leq t'$:
\begin{enumerate}[label={(\arabic*)}]
	\item $S_t$ maps to $n(t)$, $b_t = r_t$, and $j_t$ maps to $i_t$. 
	\item $(S_{t + 1}, b_{t + 1}) \mid (S_t, b_t, j_t)$ has the same distribution as $(n(t + 1), r_{t + 1}) \mid (n(t), r_t, i_t)$.
\end{enumerate}
If we can establish both $(1)$ and $(2)$ then for any sequence of states and actions the distribution on outcomes after some number of steps $k$ is the same in both models. The following theorem proves the reduction by constructing a mapping from types to classes that satisfies these criteria.
\generalreduction
\begin{proof}
Let $L = C^2$ and let $\eta = \beta$. For each type $j \in \{ 1, 2, \dots, C^2 \}$ we define a class $i = j$. (Even though $i = j$, we use separate names to distinguish between the class and type.) By this mapping, for a state $S_t = (X_1, \dots, X_{C^2})$ and vector $N(t) = (N_1, \dots, N_L)$, if $(X_1, \dots, X_{C^2}) = (N_1, \dots, N_L)$ then $S_t$ maps to $N(t)$. First we define the mapping, and then we show the reduction. 

The idea for the mapping is to construct a class $i$ for each type $j$ such that the joint distribution on reward and new arms associated with $i$ maps to the joint distribution on benefit and new children associated with $j$. 

Recall that the joint distribution on the benefit and multiset of children associated with type $j$ is
\begin{center}
$(B(j), (Y_{j1}, \dots, Y_{jC^2})) =
\left\{
	\begin{array}{ll}
		\mbox{w.p. } \overline{p}(j) & (\overline{b}(j), (U_{j1}, \dots, U_{jC^2}) \sim \overline{D}_j ) \\
		& \\
		\mbox{w.p. } 1 - \overline{p}(j) & (0, \emptyset) \\
	\end{array}
\right.$
\end{center}
While so far we have referred to $\overline{D}_j$ as a distribution on multisets of types, $\overline{D}_j$ is simply a distribution on multisets of elements in $\{ 1, 2, \dots, C^2 \}$. Since $L = C^2$, it can also be viewed as a distribution on multisets of classes $\{ 1, 2, \dots, L \}$. We define class $i$ by the joint distribution
\begin{center}
$(R(i), (N_{i1}, \dots, N_{iL})) =
\left\{
	\begin{array}{ll}
		\mbox{w.p. } \overline{p}(j) & (\overline{b}(j), (M_{j1}, \dots, M_{jL}) \sim \overline{D}_j ) \\
		& \\
		\mbox{w.p. } 1 - \overline{p}(j) & (0, \emptyset) \\
	\end{array}
\right.$
\end{center}
By this construction, if $i = j$ then $(B(j), (Y_{j1}, \dots, Y_{jC^2})) \sim (R(i), (N_{i1}, \dots, N_{iL}))$.

For any $t \geq 0$, let $(S_t, b_t)$ be an arbitrary state in the general model, and suppose the tracer chooses type $j \in S_t$. Then $b_{t + 1} = b_t + B(j) e^{-\beta t}$ and $S_{t + 1} = (S_t \setminus j) \cup (Y_{j1}, \dots, Y_{jC^2})$. Let $(n(t), r_t) = (S_t, b_t)$ be the corresponding state in the branching bandit model, and let the agent choose the corresponding class $i = j$. Then $r_{t + 1} = r_t + R(i) e^{-\eta t}$ and $n(t + 1) = (n(t) \setminus i) \cup (N_{i1}, \dots, N_{iL})$. Because $i = j$, $(B(j), (Y_{j1}, \dots, Y_{jC^2})) \sim (R(i), (N_{i1}, \dots, N_{iL}))$, and combined with the fact that $n(t) = S_t$, $\eta = \beta$, and $r_t = b_t$,
\begin{align*}
(S_{t + 1}, b_{t + 1}) \mid(S_t, b_t, j) &\sim ((S_t \setminus j) \cup (Y_{j1}, \dots, Y_{jC^2}), b_t + B(j) e^{-\beta t} ) \mid (S_t, b_t, j) \\
&\sim ((n(t) \setminus i) \cup (N_{i1}, \dots, N_{iL}), r_t + R(i) e^{-\eta t}) \mid (n(t), r_t, i) \\
&\sim (n(t + 1), r_{t + 1}) \mid (n(t), r_t, i)
\end{align*}
Therefore, given corresponding states $(S_t, b_t) = (n(t), r_t)$ and actions $j_t = i_t$, then 
$$(S_{t + 1}, b_{t + 1}) \mid (S_t, b_t, j_t) \sim (n(t + 1), r_{t + 1}) \mid (n(t), r_t, i_t).$$ 
Since the mapping preserves this next-state transition, by induction there is a sequence of corresponding states and actions in the branching bandit model, so the reduction holds.
\end{proof}

%% file: sections/discussion.tex
\section{Discussion} \label{section:discussion}
We present contact tracing as an algorithm design problem where the objective is to develop a policy that prioritizes contacts to trace. Through a clean connection to the branching bandit model~\cite{weiss}, we develop provably optimal policies for a variety of infection models. Analyzing the structure of these policies leads to qualitative insights about trade-offs in contact tracing applications. In a setting where the probability of transmission is constant, the policy reflects the prioritization of contacts by recency that we see in practice. In the more complex models of infection we study, the structure of the optimal policy depends on the specific parameters of the model yet still exhibits clear structure. Finally, we conclude with a general model, which has the capacity to model arbitrary interactions between individuals based on factors like an individual's profession, role within a community, or their risk of infection.

There are many compelling questions to consider going forward. One interesting question is how to analyze optimal policies in a dynamic setting, where contact tracing proceeds while the infection is spreading. We are currently exploring this setting in ongoing work~\cite{dynamic}. A key question here is how to choose the objective function. On the one hand, the algorithm ought to be rewarded for identifying infected cases. On the other hand, the algorithm ought to prevent the spread of new cases, which in a formal sense works against the goal of identifying infected cases. This makes defining the objective function somewhat subtle. Another interesting question to consider is backward contact tracing. In this paper we examined forward contact tracing, which traces any infections due to an individual, while in backward contact tracing the goal is to identify the source of the infection. Finally, it is an interesting question to ask how theoretical recommendations in general can be integrated into manual contact tracing.

%% file: sections/acknowledgements.tex
\section*{Acknowledgments}
Michela Meister was supported by the Department of Defense (DoD) through the National Defense Science \& Engineering Graduate (NDSEG) Fellowship Program. Jon Kleinberg was supported in part by a Simons Investigator Award, a Vannevar Bush Faculty Fellowship, MURI grant W911NF-19-0217, AFOSR grant FA9550-19-1-0183, ARO grant W911NF19-1-0057, and a grant from the MacArthur Foundation.

%% file: sections/appendix.tex
\section{Appendix: Bivariate} \label{app:biv}
Recall that each node in the bivariate model is associated with a recency $h \in \{ 0, 1, \dots, T \}$ and span $\Delta \in \{ 0, 1, \dots, T \}$, summarized by a type $(h, \Delta) \in \{ 0, 1, \dots, T \}^2$. A node of span $\Delta$ has probability of infection $p(\Delta) = p_T e^{-\alpha \Delta}$. Policies in the bivariate model are on the set of all types $\{ 0, 1, \dots, T \}^2$.

\paragraph{Defining the optimal policy.} 
Recall that the bivariate model is identical to the basic model except for the probability of infection. Therefore we can construct the optimal policy by modifying the dynamic program from~\cref{sub:basic:defining}. The element $\sigma_0$ is assigned to the type with the highest expected immediate benefit.
\begin{align}
\sigma_0 &= \mathop{\arg\max}_{(h, \Delta) \in \{ 0, 1, \dots, T \}^2 } \Exp[b((h, \Delta), \emptyset)] \label{eqn:biv:dp0_start}
\end{align} 
Given the first $k$ elements $\sigma_0, \sigma_1, \dots, \sigma_{k - 1}$, $\sigma_k$ is selected by comparing $((h, \Delta), \sigma_{k-1})$-periods across all types not already in the prefix.
\begin{align}
\sigma_k &= \mathop{\arg\max}_{(h, \Delta) \in \{ 0, 1, \dots, T \}^2 \setminus \{ \sigma_0, \dots, \sigma_{k - 1} \} } \frac{\Exp[b((h, \Delta), \sigma_{k - 1})]}{1 - \gamma((h, \Delta), \sigma_{k - 1})} \label{eqn:biv:dpk_start}
\end{align}

\paragraph{Analyzing periods.} 
Recall that a node with recency $h$ has children $Y(h) = (Y_{0, h}, Y_{1, h-1}, \dots, Y_{h-1, 1}) \sim D^h$, where $Y_{j,i}$ indicates the number of children with recency $j$ and span $i$. For any prefix $\sigma_0, \sigma_1, \dots, \sigma_{k-1}$, an $((h, \Delta), \sigma_{k - 1})$-period consists of an $((h, \Delta), \emptyset)$-period followed by a $(Y(h), \sigma_{k-1})$-epoch for $Y(h) \sim D^h$. Therefore,
\begin{align}
\Exp[b((h, \Delta), \sigma_{k-1})] &= p(\Delta) \left( e^{-\beta h} + e^{-\beta} \mathop{\Exp}_{Y(h) \sim D^h} \left[ b(Y(h), \sigma_{k-1}) \right] \right) \label{eqn:biapp:id1} \\
\gamma((h, \Delta), \sigma_{k-1}) &= e^{-\beta} \left( p(\Delta) \mathop{\Exp}_{Y(h) \sim D^h} \left[ e^{-\beta \tau(Y(h), \sigma_{k-1})} \right] + 1 - p(\Delta) \right) \label{eqn:bivapp:id2}
\end{align}

\paragraph{Restating the optimal policy.}
We restate the optimal policy using the above identities.
\begin{align}
\sigma_0 &= \mathop{\arg \max}_{(h, \Delta) \in \{ 0, 1, \dots, T \}^2 } \Exp[b((h, \Delta), \emptyset)] \nonumber \\
&= \mathop{\arg \max}_{(h, \Delta) \in \{ 0, 1, \dots, T \}^2 } p(\Delta) e^{-\beta h} \label{eqn:biv:dp0}
\end{align} 
Given the first $k$ elements $\sigma_0, \sigma_1, \dots, \sigma_{k - 1}$, $\sigma_k$ is selected from the elements not already in the prefix.
\begin{align}
\sigma_k &= \mathop{\arg \max}_{(h, \Delta) \in \{ 0, 1, \dots, T \}^2 \setminus \{ \sigma_0, \dots, \sigma_{k - 1} \} } \frac{\Exp[b((h, \Delta), \sigma_{k - 1})]}{1 - \gamma((h, \Delta), \sigma_{k - 1})} \nonumber \\
&= \mathop{\arg \max}_{(h, \Delta) \in \{ 0, 1, \dots, T \}^2 \setminus \{ \sigma_0, \dots, \sigma_{k - 1} \} } \frac{ p(\Delta) \left( e^{-\beta h} + e^{-\beta} \mathop{\Exp}_{Y(h) \sim D^h} \left[ b(Y(h), \sigma_{k-1}) \right] \right) }{1 - e^{-\beta} \left( p(\Delta) \mathop{\Exp}_{Y(h) \sim D^h} \left[ e^{-\beta \tau(Y(h), \sigma_{k-1})} \right] + 1 - p(\Delta) \right)} \label{eqn:biv:dpk}
\end{align}

\subsection{Ordering types of the same recency} \label{subsection:same_recency}
Using similar techniques as in~\cref{section:univariate}, we present our main result for the bivariate section.
\setcounter{section}{7}
\bivspan
\setcounter{section}{10}
\begin{proof}
Fix $T \in \NN$, $p_T \in (0, 1]$, and $\beta > 0$. Let $\sigma$ be the optimal priority ordering constructed via the dynamic program described above. Fix $h \in \{ 0, 1, \dots, T \}$. It suffices to show that $(h, 0) \preceq (h, 1) \preceq \dots \preceq (h, T)$. 

By~\cref{eqn:biv:dp0}
\begin{align*}
\sigma_0 &= \mathop{\arg\max}_{ (h, \Delta) \in \{ 0, 1, \dots, T \}^2 } \Exp[b(h, \emptyset)] \\
&= \mathop{\arg\max}_{ (h, \Delta) \in \{ 0, 1, \dots, T \}^2 } p_T e^{-\alpha \Delta - \beta h} \\
&= (0, 0)
\end{align*}

Since $\alpha > 0$ and $\beta > 0$, $(0, 0)$ attains the maximum. It suffices to show that for any $(h, \Delta_1) \neq (0, 0)$ and $(h, \Delta_2) \neq (0, 0)$ if $\Delta_1 < \Delta_2$ then $(h, \Delta_1) \prec (h, \Delta_2)$, since then by induction $(h, 0) \preceq (h, 1) \preceq \dots \preceq (h, T)$. 

Fix any $(h, \Delta_1) \neq (0, 0)$ and $(h, \Delta_2) \neq (0, 0)$. Assume to the contrary that $(h, \Delta_2) \prec (h, \Delta_1)$. Then $(h, \Delta_2)$ was selected prior to $(h, \Delta_1)$ in the construction of $\sigma$. That is, there is some pair $k, k'$ with $k < k'$ such that $\sigma_k = (h, \Delta_2)$ and $\sigma_{k'} = (h, \Delta_1)$. As a result, for the  prefix $\sigma_0, \dots, \sigma_{k - 1}$,
$$\frac{\Exp[b((h, \Delta_2), \sigma_{k - 1})]}{1 - \gamma((h, \Delta_2), \sigma_{k - 1})} \geq \frac{\Exp[b((h, \Delta_1), \sigma_{k - 1})]}{1 - \gamma((h, \Delta_1), \sigma_{k - 1})}$$

Given the prefix $\sigma_0, \sigma_1, \dots, \sigma_{k - 1}$, define the expected benefit and pre-multiplier associated with a $(Y(h), \sigma_{k - 1})$-epoch for $Y(h) \sim D^h$. 
\begin{align*}
b_h &= \mathop{\Exp}_{Y(h) \sim D^h} [b(Y(h), \sigma_{k - 1})] \\
\gamma_h &= \mathop{\Exp}_{Y(h) \sim D^h} [\gamma(Y(h), \sigma_{k - 1})]
\end{align*}

Applying these identities to~\cref{eqn:biapp:id1,eqn:bivapp:id2}, for any $\Delta \in \{ 0, 1, \dots, T \}$,
\begin{align*}
\frac{\Exp[b((h, \Delta), \sigma_{k-1})]}{1 - \gamma((h, \Delta), \sigma_{k-1})} = \frac{ p_T e^{-\alpha \Delta } \left( e^{-\beta h} + e^{-\beta} b_h \right) }{1 - e^{-\beta} + e^{-\beta} (1 - \gamma_h) p_T e^{-\alpha \Delta } }
\end{align*}

Define the function $f_h : \RR \rightarrow \RR$ as
\begin{align*}
f_h(x) &=  \frac{ p_T e^{-\alpha x } \left( e^{-\beta h} + e^{-\beta} b_h \right) }{1 - e^{-\beta} + e^{-\beta} (1 - \gamma_h) p_T e^{-\alpha x } } 
\end{align*}
Then 
\begin{align*}
f_h(\Delta_1) &= \frac{\Exp[b((h, \Delta_1), \sigma_{k - 1})]}{1 - \gamma((h, \Delta_1), \sigma_{k - 1})} \\
f_h(\Delta_2) &= \frac{\Exp[b((h, \Delta_2), \sigma_{k - 1})]}{1 - \gamma((h, \Delta_2), \sigma_{k - 1})}
\end{align*}
Since $\sigma_k = (h, \Delta_2)$, $f_h(\Delta_2) \geq f_h(\Delta_1)$.
For strictly positive constants 
\begin{align*}
C_0 &= p_T \left( e^{-\beta h} + e^{-\beta} b_h \right) \\
C_1 &= 1 - e^{-\beta} \\
C_2 &= p_T e^{-\beta} (1 - \gamma_h),
\end{align*}
we can rewrite $f_h$ as
\begin{align*}
f_h(x) &= \frac{ C_0 e^{-\alpha x} }{C_1 + C_2 e^{-\alpha x } }
\end{align*}
Then 
\begin{align*}
f_h'(x) &= \frac{ - C_0 C_1 \alpha e^{-\alpha x} } { (C_1 + C_2 e^{-\alpha x})^2 }
\end{align*}
Since $C_0, C_1, C_2 > 0$, $f_h'(x) < 0$ for all $x \in \RR$. Therefore $f_h$ is monotonically decreasing. This implies $f_h(\Delta_2) < f_h(\Delta_1)$, which contradicts the assumption that $f_h(\Delta_2) \geq f_h(\Delta_1)$. Thus no such prefix $\sigma_0, \sigma_1, \dots, \sigma_{k - 1}$ exists, so $(h, \Delta_1) \preceq (h, \Delta_2)$, and as a result $(h, 0) \preceq (h, 1) \preceq \dots \preceq (h, T)$.
\end{proof}

\subsection{Ordering types of the same span} \label{subsection:same_span}
In this section we examine ordering types with the same span. Recall that two nodes of different recencies have no children of the same type. This makes comparing periods and epochs tricky, and as a result, this section uses fairly different techniques. In order to determine whether $(h, \Delta) \prec (h', \Delta')$, we analyze the construction of optimal policies step by step and examine which type appears first in the ordering. Our techniques require case-by-case analysis, so we restrict our analysis to the setting in which $D$ is a Bernoulli distribution and the recencies of all nodes are limited to $\{ 0, 1, 2 \}$. 

Our main result shows that, under these restrictions, it is optimal to query nodes of the same span in order of recency.
\setcounter{section}{7}
\bivrecency
\setcounter{section}{10}
\begin{proof}
Fix $T \in \NN$, $p_T \in (0, 1]$, $c \in (0, 1]$, $\alpha > 0$, and $\beta > \ln\left( 2 (1 + c p_T e^{-\alpha})/(1 - e^{-\alpha}) \right)$. Let $D = \Ber(c)$. It suffices to show that there is an optimal policy where for any $\Delta \in \{ 0, 1, \dots, T \}$, $(0, \Delta) \preceq (1, \Delta) \preceq (2, \Delta)$.~\Cref{lem:0_delta_1_delta} shows that $(0, \Delta) \preceq (1, \Delta)$, and~\cref{lem:1_delta_2_delta} shows that $(1, \Delta) \preceq (2, \Delta)$, giving the desired result.
\end{proof}

\paragraph{Overview of techniques.} 
Recall that a node of type $(h, \Delta)$ has children defined by the multiset $Y(h) = (Y_{0, h}, Y_{1, h-1}, \dots, Y_{h - 1, 1}) \sim D^h$, where $Y_{j,i}$ indicates the number of children with recency $j$ and span $i$. For the following proofs $D$ is a Bernoulli distribution, so $Y_{j, i} \in \{ 0, 1 \}$. We say that a node $v$ with recency $h(v) = h$ has \textit{potential children} $\{ (0, h), (1, h-1), \dots, (h-1, 1) \}$, since each type in the set could be a child of $v$. Let $W(h, \Delta)$ be the set of all \textit{potential descendants} of a node of type $(h, \Delta)$.~\Cref{fig:biv_tree} shows the potential descendants of different types we examine in this section.

\begin{figure}
\begin{center}
\includegraphics[width=.5\linewidth]{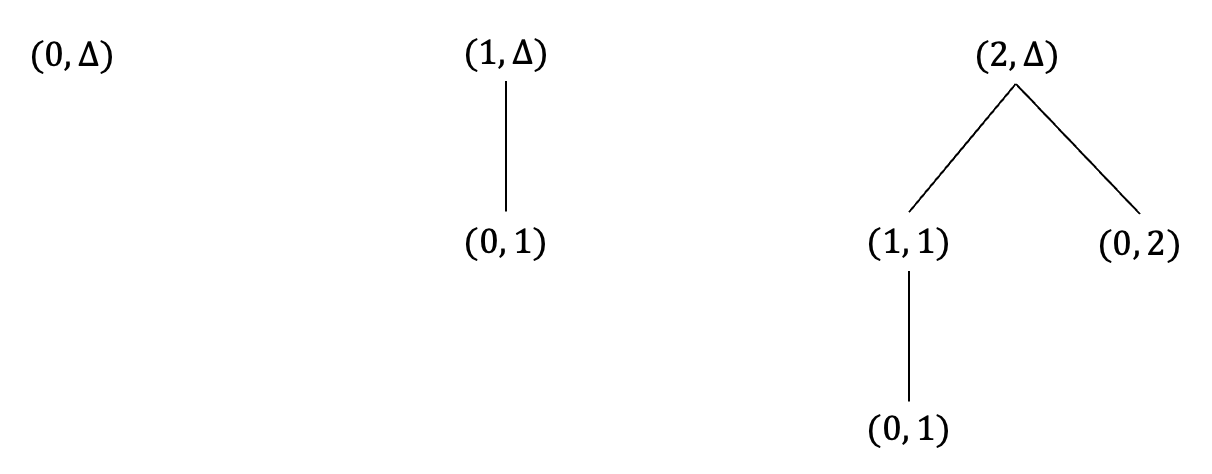}
\caption{In comparing types with the same span $\Delta$, we examine a version of the bivariate model where each individual meets at most one contact each day. A node of recency $0$ has no potential children, a node of recency $1$ has one potential child, and a node of recency $2$ has two potential children, one with a potential child of its own.}
\label{fig:biv_tree}
\end{center}
\end{figure}

Our strategy is to analyze the dynamic program from~\cref{eqn:biv:dp0_start,eqn:biv:dpk_start}. By definition, an optimal priority ordering $\sigma$ sequences all types in $\{ 0, 1, \dots, T \}^2$. However, determining whether $(h, \Delta) \preceq (h', \Delta')$ does not necessarily require computing the entirety of $\sigma$. Instead, we only need to determine which of $(h, \Delta)$ and $(h', \Delta')$ falls earlier in the ordering $\sigma$. To do this, for any prefix $\sigma_0, \sigma_1, \dots, \sigma_{k - 1}$, we need to be able to compare an $((h, \Delta), \sigma_{k - 1})$-period to an $((h', \Delta'), \sigma_{k - 1})$-period. An $((h, \Delta), \sigma_{k - 1})$-period only involves the type $(h, \Delta)$ and its potential descendants. Therefore, to decide whether $(h, \Delta) \preceq (h', \Delta')$, it is sufficient to compute an optimal ordering on only the two types in question, $(h, \Delta)$ and $(h', \Delta')$, and their potential descendants, $W(h, \Delta)$ and $W(h', \Delta')$. Specifically, we construct an ordering $\pi$ on the set of types
$$S_\pi = (h, \Delta) \cup (h', \Delta') \cup W(h, \Delta) \cup W(h', \Delta')$$
such that if $(h, \Delta) \preceq (h', \Delta')$ according to $\pi$, it follows that there is an optimal ordering $\sigma$ on $\{ 0, 1, \dots, T \}^2$ such that $\pi$ is a subsequence of $\sigma$ and therefore $(h, \Delta) \preceq (h', \Delta')$ according to $\sigma$. 

To simplify notation, for any type $(h, \Delta)$ and prefix $\pi_0, \pi_1, \dots, \pi_{k - 1}$, we define the index function
\begin{align*}
I((h, \Delta), \pi_{k - 1}) &= \frac{\Exp[b((h, \Delta), \pi_{k - 1})]}{1 - \gamma((h, \Delta), \pi_{k - 1})}
\end{align*}
Overloading notation, for any recency $h \in \{ 0, 1, \dots, T \}$ we define the benefit function
\begin{align*}
b(h) &= e^{-\beta h}
\end{align*}
The dynamic program is defined exactly as before, now with the index function notation. The highest priority element according to $\pi$ is the type in $S_\pi$ with the highest expected immediate benefit.
\begin{align*}
\pi_0 &= \mathop{\arg\max}_{ (h, \Delta) \in S_\pi } I((h, \Delta), \emptyset)
\end{align*}
Given a prefix $\pi_0, \pi_1, \dots, \pi_{k - 1}$, $\pi_k$ is chosen from the remaining elements in $S_\pi$.
\begin{align*}
\pi_k &= \mathop{\arg\max}_{ (h, \Delta) \in S_\pi \setminus \{ \pi_0, \dots, \pi_{k - 1} \} } I((h, \Delta), \pi_{k - 1})
\end{align*}

\paragraph{Proof organization.}
The proof is split into two main lemmas.~\Cref{lem:0_delta_1_delta} shows that $(0, \Delta) \preceq (1, \Delta)$, and~\cref{lem:1_delta_2_delta} shows that $(1, \Delta) \preceq (2, \Delta)$.~\Cref{lem:0_delta_1_delta} serves as an introduction to constructing optimal priority orderings in a fairly simple setting where $\pi$ involves only three types, and ~\cref{lem:1_delta_2_delta} extends these techniques to handle a larger set of types. Because of this added complexity, the proof of~\cref{lem:1_delta_2_delta} is split into cases. To show that $(1, \Delta) \preceq (2, \Delta)$,~\cref{lem:10_20,lem:11_21} handle the cases when $\Delta \in \{ 0, 1 \}$ and~\cref{lem:first_prefix,lem:second_prefix,lem:third_prefix} handle the cases when $\Delta > 1$. 

\subsubsection{Proving $(0, \Delta) \preceq (1, \Delta)$}
\begin{lemma} \label{lem:0_delta_1_delta}
Fix $T \in \NN$, $\Delta \in \{ 0, 1, \dots, T \}$, $p_T \in (0, 1]$, $c \in (0, 1]$, $\alpha > 0$, and $\beta > \ln(1 + c p_T e^{-\alpha} - c p_T e^{-\alpha \Delta})$. Let $D = \Ber(c)$. Then $(0, \Delta) \preceq (1, \Delta)$.
\end{lemma}
\begin{proof}
Recall that our strategy is to construct an ordering $\pi$ that is a subsequence of an optimal ordering $\sigma$. Consider $\Delta = 0$. Then $S_\pi = \{ (0, 0), (1, 0), (0, 1) \}$, so
\begin{align*}
\pi_0 &= \mathop{\arg \max}_{(h, \Delta) \in S_\pi} I((h, \Delta), \emptyset) = \mathop{\arg \max}_{(h, \Delta) \in S_\pi} \frac{ p_T e^{-\alpha \Delta - \beta h} }{ 1 - e^{-\beta} } = (0, 0)
\end{align*}
Thus $(0, 0) \prec (0, 1)$. 

Consider $\Delta = 1$. Then $S_\pi = \{ (0, 1), (1, 1) \}$, so
\begin{align*}
\pi_0 &= \mathop{\arg \max}_{(h, \Delta) \in S_\pi} I((h, \Delta), \emptyset) = \mathop{\arg \max}_{(h, \Delta) \in S_\pi} \frac{ p_T e^{-\alpha \Delta - \beta h} }{ 1 - e^{-\beta} } = (0, 1)
\end{align*}
Thus $(0, 1) \prec (1, 1)$.
 
Now consider $\Delta > 1$. Then $S_\pi = \{ (0, \Delta), (1, \Delta), (0, 1) \}$, so
\begin{align*}
\pi_0 &= \mathop{\arg \max}_{(h, \Delta) \in S_\pi} I((h, \Delta), \emptyset) = \mathop{\arg \max}_{(h, \Delta) \in S_\pi} \frac{ p_T e^{-\alpha \Delta - \beta h} }{ 1 - e^{-\beta} } = (0, 1)
\end{align*}
Here, since $\Delta > 1$, $(0, 1)$ maximizes the index function. To select $\pi_1$ we compare $I((0, \Delta), \pi_0)$ and $I((1, \Delta), \pi_0)$. The type $(0, \Delta)$ has no potential children, so 
\begin{align}
I((0, \Delta), \pi_0) &= I((0, \Delta), \emptyset) \nonumber \\
&= \frac{ p(\Delta) b(0)}{1 - e^{-\beta}} \label{eq:0_delta_pv}
\end{align}
The type $(1, \Delta)$ has potential child $(0, 1)$. A $((1, \Delta), \pi_0)$-period begins by querying a node of type $(1, \Delta)$ on step $t = 0$. With probability $p(\Delta)$, the node is infected, and with probability $c$ it has child $(0, 1)$. So with probability $c p(\Delta)$ the child $(0, 1)$ is queried on step $t = 1$, otherwise the period ends on step $t = 0$. Therefore 
\begin{align}
I((1, \Delta), \pi_0) &= \frac{ \Exp[b((1, \Delta), \pi_0)]}{1 - \gamma((1, \Delta), \pi_0)} \nonumber \\
&= \frac{ p(\Delta) (b(1) + c p(1) b(0) e^{-\beta}) }{1 - [(1 - c p(\Delta)) e^{-\beta} + c p(\Delta) e^{-2 \beta}]} \nonumber \\
&= \frac{ p(\Delta) (b(1) + c p(1) b(0) e^{-\beta}) }{1 - e^{-\beta} + c p(\Delta) e^{-\beta} - c p(\Delta) e^{-2 \beta}} \nonumber \\
&= \frac{ p(\Delta) (b(1) + c p(1) b(0) e^{-\beta}) }{(1 - e^{-\beta})(1 + c p(\Delta) e^{-\beta})} \label{eq:1_delta_pv}
\end{align}
Combining~\cref{eq:0_delta_pv,eq:1_delta_pv} we have
\begin{align*}
I((0, \Delta), \pi_0) - I((1, \Delta), \pi_0) &= \frac{ p(\Delta) b(0)}{1 - e^{-\beta}} - \frac{ p(\Delta) (b(1) + c p(1) b(0) e^{-\beta}) }{(1 - e^{-\beta})(1 + c p(\Delta) e^{-\beta})} \\
&= \frac{ p(\Delta)}{1 - e^{-\beta}} \cdot \left( b(0) - \frac{b(1) + c p(1) b(0) e^{-\beta} }{1 + c p(\Delta) e^{-\beta}} \right) \\
&= \frac{ p_T e^{-\alpha \Delta} }{1 - e^{-\beta}} \cdot \left( 1 - \frac{e^{-\beta} + c p_T e^{-\alpha - \beta} }{1 + c p_T e^{-\alpha \Delta - \beta} } \right)
\end{align*}
The above expression is strictly positive iff 
\begin{align*}
1 &> \frac{e^{-\beta} + c p_T e^{-\alpha - \beta} }{1 + c p_T e^{-\alpha \Delta - \beta}} 
\end{align*}
Setting $\beta >\ln(1 + c p_T e^{-\alpha} - c p_T e^{-\alpha \Delta})$ satisfies the above inequality, so $I((0, \Delta), \pi_0) > I((1, \Delta), \pi_0)$. Thus $\pi_1 = (0, \Delta)$, and thus for $\Delta > 1$, $(0, \Delta) \prec (1, \Delta)$. Therefore for any $\Delta \in \{ 0, 1, \dots, T \}$, $(0, \Delta) \prec (1, \Delta)$.
\end{proof}

\subsubsection{Proving $(1, \Delta) \preceq (2, \Delta)$}
The following lemma analyzes an optimal priority ordering case by case to show that in all cases for the given regime, $(1, \Delta) \preceq (2, \Delta)$.
\begin{lemma} \label{lem:1_delta_2_delta}
Fix $T \in \NN$, $\Delta \in \{ 0, 1, \dots, T \}$, $p_T \in (0, 1]$, $c \in (0, 1]$, $\alpha > 0$, and $\beta > \ln\left( 2 (1 + c p_T e^{-\alpha})/(1 - e^{-\alpha}) \right)$. Let $D = \Ber(c)$. Then $(1, \Delta) \preceq (2, \Delta)$.
\end{lemma}
\begin{proof}
Recall that the proof idea is to construct an ordering $\pi$ on a subset of types, where $\pi$ is a subsequence of an optimal ordering $\sigma$ on $\{ 0, 1, \dots, T \}^2$. Specifically, let $S_\pi = \{ (0, 1), (1, 1), (0, 2), (1, \Delta), (2, \Delta) \}$. We will analyze a priority ordering $\pi$ on $S_\pi$. The cases for $\Delta = 0$ and $\Delta = 1$ are covered by~\cref{lem:10_20} and~\cref{lem:11_21}, respectively. We analyze the case when $\Delta > 1$ in multiple parts, as follows.

Fix $\Delta \in \{ 2, \dots, T \}$. Then, as in~\cref{lem:0_delta_1_delta,lem:10_20,lem:11_21}, $\pi_0 = (0, 1)$. Next we examine $\pi_1$. By~\cref{thm:biv:span} $(1, 1) \preceq (1, \Delta)$, and by~\cref{lem:11_21} and~\cref{thm:biv:span} $(1, 1) \preceq (2, 1) \preceq (2, \Delta)$. Therefore neither $(1, \Delta)$ nor $(2, \Delta)$ can be assigned to $\pi_1$, leaving $(1, 1)$ and $(0, 2)$ as the only remaining candidates. We now examine both cases and show that regardless of whether $\pi_1 = (1, 1)$ or $\pi_1 = (0, 2)$, $(1, \Delta) \preceq (2, \Delta)$. 

Suppose $\pi_1 = (0, 2)$. Consider selecting $\pi_2$ from the set of remaining types $\{ (1, 1), (1, \Delta), (2, \Delta) \}$. Following the same argument as described above, $(1, 1)$ precedes both $(1, \Delta)$ and $(2, \Delta)$, so $\pi_2 = (1, 1)$. For the given prefix, by~\cref{lem:first_prefix} $I((1, \Delta), \pi_2) > I((2, \Delta), \pi_2)$ so $\pi_3 = (1, \Delta)$, and therefore $(1, \Delta) \preceq (2, \Delta)$.

Suppose $\pi_1 = (1, 1)$. Consider selecting $\pi_2$ from the set of remaining types $\{ (0, 2), (1, \Delta), (2, \Delta) \}$. By~\cref{lem:second_prefix} $I((1, \Delta), \pi_1) > I((2, \Delta), \pi_1)$, so $\pi_2 \neq (2, \Delta)$. If $\pi_2 = (1, \Delta)$, then $(1, \Delta) \preceq (2, \Delta)$. Suppose instead that $\pi_2 = (0, 2)$. For the given prefix, by~\cref{lem:third_prefix} $I((1, \Delta), \pi_2) > I((2, \Delta), \pi_2)$, so $\pi_3 = (1, \Delta)$. Thus with the given bound on $\beta$, in any optimal priority ordering $(1, \Delta) \preceq (2, \Delta)$.
\end{proof}
\Cref{lem:10_20,lem:11_21} follow similar structures, and the relatively minor differences between the two proofs are due to the change in span. 
\begin{lemma} \label{lem:10_20}
Fix $T \in \NN$, $p_T \in (0, 1]$, $c \in (0, 1]$, and $\alpha, \beta > 0$. Let $D = \Ber(c)$. Then $(1, 0) \preceq (2, 0)$.
\end{lemma}
\begin{proof}
Let $S_\pi = \{ (0, 1), (1, 0), (1, 1), (0, 2), (2, 0) \}$. Since $(1, 0)$ and $(1, 1)$ both have recency $1$, by~\cref{thm:biv:span} it is optimal to query $(1, 0)$ before $(1, 1)$. Therefore, there is an optimal ordering $\sigma$ where $(1, 0) \preceq (1, 1)$. Then for some pair $k, k'$, with $k' > k$, $\pi_k = (1, 0)$ and $\pi_{k'} = (1, 1)$. Since $(1, 1) \notin \{ \pi_0, \dots, \pi_k \}$, when computing $I((2, 0), \pi_j)$ for any $0 \leq j \leq k$ the potential child $(1, 1)$ is never queried. Therefore we can ignore the potential child $(1,1)$ when computing $\pi$, since it does not affect which of $(1, 0)$ or $(2, 0)$ appears first. Instead, for the purposes of this particular proof, we can think of $(2, 0)$ as having a single potential child $(0, 2)$. Thus we only need to choose an ordering on the types $S_\pi' = \{ (0, 1), (1, 0), (0, 2), (2, 0) \}$.

Consider $\alpha \geq \beta$. Then the type $(1, 0)$ (which recall has recency $h = 1$ and span $\Delta = 0$) is assigned to $\pi_0$.
\begin{align*}
\pi_0 &= \mathop{\arg \max}_{(h, \Delta) \in S_\pi'} I((h, \Delta), \emptyset) = \mathop{\arg \max}_{(h, \Delta) \in S_\pi'} \frac{ p_T e^{-\alpha \Delta - \beta h} }{ 1 - e^{-\beta} } = (1, 0)
\end{align*}
Thus if $\alpha \geq \beta$, $(1, 0) \preceq (2, 0)$.

Consider $\alpha < \beta$. Then 
\begin{align*}
\pi_0 &= \mathop{\arg \max}_{(h, \Delta) \in S_\pi'} I((h, \Delta), \emptyset) = \mathop{\arg \max}_{(h, \Delta) \in S_\pi'} \frac{ p_T e^{-\alpha \Delta - \beta h} }{ 1 - e^{-\beta} } = (0, 1)
\end{align*}
Given the prefix $\pi_0$, we now compute $\pi_1$. The possible candidates for $\pi_1$ are $(1, 0)$, $(2, 0)$ and $(0, 2)$. Since neither $(2, 0)$ nor $(0, 2)$ has $\pi_0$ as a child, $I((2, 0), \pi_0) = I((2, 0), \emptyset)$ and $I((0, 2), \emptyset) = I((0, 2), \emptyset)$. Since $\alpha < \beta$, 
\begin{align*}
I((0, 2), \pi_0) = I((0, 2), \emptyset) > I((2, 0), \pi_0) = I((2, 0), \emptyset)
\end{align*}
Therefore $\pi_1 \neq (2, 0)$, leaving $(1, 0)$ and $(0, 2)$ as the remaining candidates. If $\pi_1 = (1, 0)$, then $(1, 0) \preceq (2, 0)$, and we are done. Suppose $\pi_1 = (0, 2)$. Now we compute $\pi_2$ by comparing $I((1, 0), \pi_1)$ and $I((2,0), \pi_1)$. The type $(1, 0)$ has potential child $(0, 1)$ in the prefix and the type $(2, 0)$ has potential child $(0, 2)$ in the prefix. Mirroring the process described in~\cref{lem:0_delta_1_delta}, 
\begin{align*}
I((1, 0), \pi_1) &= \frac{ \Exp[b((1, 0), \pi_1)]}{1 - \gamma((1, 0), \pi_1)} \\
&= \frac{ p(0) (b(1) + c p(1) b(0) e^{-\beta}) }{1 - [(1 - c p(0)) e^{-\beta} + c p(\Delta) e^{-2 \beta}]} \\
&= \frac{ p(0) (b(1) + c p(1) b(0) e^{-\beta}) }{1 - e^{-\beta} + c p(0) e^{-\beta} - c p(0) e^{-2 \beta}} \\
&= \frac{ p(0) (b(1) + c p(1) b(0) e^{-\beta}) }{(1 - e^{-\beta})(1 + c p(0) e^{-\beta})}
\end{align*}
\begin{align*}
I((2, 0), \pi_1) &= \frac{ \Exp[b((2, 0), \pi_1)]}{1 - \gamma((2, 0), \pi_1)} \\
&= \frac{ p(0) (b(2) + c p(2) b(0) e^{-\beta}) }{1 - [(1 - c p(0)) e^{-\beta} + c p(\Delta) e^{-2 \beta}]} \\
&= \frac{ p(0) (b(2) + c p(2) b(0) e^{-\beta}) }{1 - e^{-\beta} + c p(0) e^{-\beta} - c p(0) e^{-2 \beta}} \\
&= \frac{ p(0) (b(2) + c p(2) b(0) e^{-\beta}) }{(1 - e^{-\beta})(1 + c p(0) e^{-\beta})}
\end{align*}
Then we have
\begin{align}
I((1, 0), \pi_1) - I((2, 0), \pi_1) &= \frac{ p(0) (b(1) + c p(1) b(0) e^{-\beta}) }{(1 - e^{-\beta})(1 + c p(0) e^{-\beta})} - \frac{ p(0) (b(2) + c p(2) b(0) e^{-\beta}) }{(1 - e^{-\beta})(1 + c p(0) e^{-\beta})} \label{eq:10_20_pv} \\
&= \frac{ p(0) }{ (1 - e^{-\beta})(1 + c p(0) e^{-\beta}) } \cdot \left( (b(1) + c p(1) b(0) e^{-\beta}) -  (b(2) + c p(2) b(0) e^{-\beta}) \right) \nonumber \\
&= \frac{ p_T }{ (1 - e^{-\beta})(1 + c p_T e^{-\beta}) } \cdot \left( e^{-\beta} + c p_T e^{-\alpha - \beta} - e^{-2 \beta} - c p_T e^{-2 \alpha - \beta}  \right) \nonumber \\
&= \frac{ p_T e^{-\beta} }{ (1 - e^{-\beta})(1 + c p_T e^{-\beta}) } \cdot \left( 1 + c p_T e^{-\alpha} - e^{- \beta} - c p_T e^{-2 \alpha }  \right) \nonumber \\
&= \frac{ p_T e^{-\beta} }{ (1 - e^{-\beta})(1 + c p_T e^{-\beta}) } \cdot \left( 1 - e^{-\beta} + c p_T e^{-\alpha}( 1 - e^{-\alpha}) \right) \nonumber \\
&> 0
\end{align}
The final inequality comes from the fact that $\alpha, \beta > 0$. Thus $I((1, 0), \pi_1) > I((2, 0), \pi_1)$, so $\pi_2 = (1, 0)$, and therefore $(1, 0) \prec (2, 0)$.
\end{proof}
The next lemma follows a similar proof structure, this time for $\Delta = 1$.
\begin{lemma} \label{lem:11_21}
Fix $T \in \NN$, $p_T \in (0, 1]$, $c \in (0, 1]$, $\alpha > 0$, and $\beta > \ln(1 + c p_T e^{-\alpha})$. Let $D = \Ber(c)$. Then $(1, 1) \preceq (2, 1)$.
\end{lemma}
\begin{proof}
Let $S_\pi = \{ (0, 1), (1, 1), (0, 2), (2, 1) \}$. Then
\begin{align*}
\pi_0 &= \mathop{\arg \max}_{(h, \Delta) \in S_\pi} I((h, \Delta), \emptyset) = \mathop{\arg \max}_{(h, \Delta) \in S_\pi} \frac{ p_T e^{-\alpha \Delta - \beta h} }{ 1 - e^{-\beta} } = (0, 1)
\end{align*}
To compute $\pi_1$, we compare $I((1, 1), \pi_0)$, $I((2, 1), \pi_0)$, and $I((0, 2), \pi_0)$. The type $(2, 1)$ has no potential children in the prefix, so
\begin{align*}
I((2, 1), \pi_0) &= I((2, 1), \emptyset) \\
&= \frac{ p(1) b(2) }{ 1 - e^{-\beta} } \\
&= \frac{ p_T e^{-\alpha - 2 \beta } }{ 1 - e^{-\beta} }
\end{align*}
The type $(1, 1)$ has potential child $(0, 1)$ in the prefix. By~\cref{eq:1_delta_pv} in~\cref{lem:0_delta_1_delta}, 
\begin{align*}
I((1, 1), \pi_0) &= \frac{ p(1) (b(1) + c p(1) b(0) e^{-\beta}) }{(1 - e^{-\beta})(1 + c p(1) e^{-\beta})} \\
&= \frac{ p_T e^{-\alpha} (e^{-\beta} + c p_T e^{-\alpha - \beta}) }{(1 - e^{-\beta})(1 + c p_T e^{-\alpha - \beta} )}
\end{align*}
Comparing $I((1, 1), \pi_0)$ and $I((2, 1), \pi_0)$, we have
\begin{align*}
I((1, 1), \pi_0) - I((2, 1), \pi_0) &= \frac{ p_T e^{-\alpha} (e^{-\beta} + c p_T e^{-\alpha - \beta}) }{(1 - e^{-\beta})(1 + c p_T e^{-\alpha - \beta} )} - \frac{ p_T e^{-\alpha - 2 \beta } }{ 1 - e^{-\beta} } \\
&= \frac{ p_T e^{-\alpha - \beta}}{ 1 - e^{-\beta} } \cdot \left( \frac{1 + c p_T e^{-\alpha} }{1 + c p_T e^{-\alpha - \beta} } - e^{-\beta} \right) \\
&\geq \frac{ p_T e^{-\alpha - \beta}}{ 1 - e^{-\beta} } \cdot \left( \frac{1 + c p_T e^{-\alpha} }{1 + c p_T e^{-\alpha - \beta} } - \frac{1}{1 + c p_T e^{-\alpha} } \right) \\ 
&> 0
\end{align*}
The first inequality is due to the lower bound on $\beta$. Thus $I((1, 1), \pi_0) > I((2, 1), \pi_0)$, so $\pi_1 \neq (2, 1)$, leaving $(1, 1)$ and $(0, 2)$ as the only remaining candidates for $\pi_1$. If $\pi_1 = (1, 1)$, then $(1, 1) \prec (2, 1)$, and we are done. Consider the case that $\pi_1 = (0, 2)$, and examine $\pi_2$. The two remaining candidates are $(1, 1)$ and $(2, 1)$. Since $\pi_0 = (0, 1)$ and $\pi_1 = (0, 2)$, the set up is identical to the conclusion of~\cref{lem:10_20}, substituting $\Delta = 1$ for $\Delta = 0$. Applying this substitution to~\cref{eq:10_20_pv},
\begin{align*}
I((1, 1), \pi_1) - I((2, 1), \pi_1) &= \frac{ p(1) (b(1) + c p(1) b(0) e^{-\beta}) }{(1 - e^{-\beta})(1 + c p(1) e^{-\beta})} - \frac{ p(1) (b(2) + c p(2) b(0) e^{-\beta}) }{(1 - e^{-\beta})(1 + c p(1) e^{-\beta})} \\
&= \frac{ p(1)}{(1 - e^{-\beta})(1 + c p(1) e^{-\beta})} \cdot \left((b(1) + c p(1) b(0) e^{-\beta}) - (b(2) + c p(2) b(0) e^{-\beta})\right) \\
&= \frac{ p_T e^{-\alpha} }{(1 - e^{-\beta})(1 + c p_T e^{-\alpha - \beta})} \cdot \left(e^{-\beta} + c p_T e^{-\alpha - \beta} - e^{-2 \beta} - c p_T e^{-2 \alpha - \beta} \right) \\
&= \frac{ p_T e^{-\alpha - \beta} }{(1 - e^{-\beta})(1 + c p_T e^{-\alpha - \beta})} \cdot \left(1 + c p_T e^{-\alpha} - e^{- \beta} - c p_T e^{-2 \alpha} \right) \\
&= \frac{ p_T e^{-\alpha - \beta} }{(1 - e^{-\beta})(1 + c p_T e^{-\alpha - \beta})} \cdot \left(1 - e^{- \beta} + c p_T e^{-\alpha} (1 - e^{-\alpha}) \right) \\
&> 0
\end{align*}
The final inequality comes from the fact that $\alpha, \beta > 0$. Thus $I((1, 1), \pi_1) > I((2, 1), \pi_1)$, so $\pi_2 = (1, 1)$, which implies that $(1, 1) \prec (2, 1)$.
\end{proof}
In the remaining lemmas, we examine cases where $\Delta > 1$. The following lemma simplifies many of the proofs that follow, by showing that in some regimes, it suffices to compare expected benefits.
\begin{lemma}\label{lem:gamma_lemma}
Fix $T \in \NN$, $\Delta \in \{ 0, 1, \dots, T \}$, $p_T \in (0, 1]$, $c \in (0, 1]$, and $\alpha, \beta > 0$. Let 
$$S_\pi = \{ (1, \Delta), (2, \Delta), (0, 1), (1, 1), (0, 2) \}.$$ 
Let $\pi$ be an ordering on $S_\pi$ that is a subsequence of an optimal priority ordering. Fix a prefix $\pi_0, \pi_1, \dots, \pi_k$. If the type $(2, \Delta)$ has at least one potential child in the prefix, then $\gamma((2, \Delta), \pi_k) \leq \gamma((1, \Delta), \pi_k)$.
\end{lemma}
\begin{proof}
By definition, for any type $(h, \Delta)$ and a prefix $\pi_0, \pi_1, \dots, \pi_k$,
\begin{align}
\gamma((h, \Delta), \pi_k) &= \sum_{j=1}^\infty \Pr[\tau((h, \Delta), \pi_k) = j] e^{-j \beta} \nonumber \\
&\leq \Pr[\tau((h, \Delta), \pi_k) = 1] e^{-\beta} + \Pr[\tau((h, \Delta), \pi_k) > 1] e^{- 2 \beta } \label{eq:gamma_bd}
\end{align}
The duration $\tau((h, \Delta), \pi_k)$ is strictly greater than 1 if and only if the node of type $(h, \Delta)$ is infected and has at least one realized child in the prefix. By definition, a node of type $(h, \Delta)$ is infected with probability $p(\Delta)$. Since each potential child is realized with probability $c$, if the type $(h, \Delta)$ has $l$ potential children in the prefix, a node of type $(h, \Delta)$ has at least one realized child in the prefix with probability $1 - (1 - c)^l$. Therefore
\begin{align}
\Pr[\tau((h, \Delta), \pi_k) > 1] &= p(\Delta) ( 1 - (1 - c)^l) \label{eq:child_bd}
\end{align}
The type $(1, \Delta)$ has only one potential descendant, the child $(0, 1)$, so a $((1, \Delta), \pi_k)$-period has duration at most $2$. 
\begin{align}
\gamma((1, \Delta), \pi_k) &= \Pr[\tau((1, \Delta), \pi_k) = 1] e^{-\beta} + \Pr[\tau((1, \Delta), \pi_k) = 2] e^{- 2 \beta } \nonumber \\
&= \Pr[\tau((1, \Delta), \pi_k) = 1] e^{-\beta} + \Pr[\tau((1, \Delta), \pi_k) > 1] e^{- 2 \beta } \label{eq:gamma_1}
\end{align}
Then by~\cref{eq:child_bd},
\begin{align}
\Pr[\tau((1, \Delta), \pi_k) > 1] &= p(\Delta) ( 1 - (1 - c)) \nonumber \\
&= p(\Delta) c \label{eq:p_tau1}
\end{align}
Consider the type $(2, \Delta)$. By~\cref{eq:gamma_bd},
\begin{align}
\gamma((2, \Delta), \pi_k) &\leq \Pr[\tau((2, \Delta), \pi_k) = 1] e^{-\beta} + \Pr[\tau((2, \Delta), \pi_k) > 1] e^{- 2 \beta } \label{eq:gamma_2}
\end{align}
Since the type $(2, \Delta)$ has at least one potential child in the prefix, by~\cref{eq:child_bd} and the fact that $c \in (0, 1]$,
\begin{align}
\Pr[\tau((2, \Delta), \pi_k) > 1] &= p(\Delta) ( 1 - (1 - c)^l) \nonumber \\
&\geq p(\Delta) ( 1 - (1 - c)) \nonumber \\
&= p(\Delta) c \label{eq:p_tau2}
\end{align}
By~\cref{eq:p_tau1} and~\cref{eq:p_tau2}, $\Pr[\tau((2, \Delta), \pi_k) > 1] \geq \Pr[\tau((1, \Delta), \pi_k) > 1]$. This implies that $\Pr[\tau((2, \Delta), \pi_k) = 1] \leq \Pr[\tau((1, \Delta), \pi_k) = 1]$. Then by~\cref{eq:gamma_1},~\cref{eq:gamma_2}, and the fact that $e^{-2 \beta} < e^{-\beta}$ for $\beta > 0$, 
\begin{align*}
\gamma((2, \Delta), \pi_k) &\leq \Pr[\tau((2, \Delta), \pi_k) = 1] e^{-\beta} + \Pr[\tau((2, \Delta), \pi_k) > 1] e^{- 2 \beta } \\
&\leq \Pr[\tau((1, \Delta), \pi_k) = 1] e^{-\beta} + \Pr[\tau((1, \Delta), \pi_k) > 1] e^{- 2 \beta } \\
&= \gamma((1, \Delta), \pi_k)
\end{align*}
Thus $\gamma((2, \Delta), \pi_k) \leq \gamma((1, \Delta), \pi_k)$.
\end{proof}
\Cref{lem:first_prefix,lem:second_prefix,lem:third_prefix} analyze different prefixes of an optimal priority ordering, as referenced by~\cref{lem:1_delta_2_delta}. 
\begin{lemma} \label{lem:first_prefix}
Fix $T \in \NN$, $\Delta \in \{ 2, \dots, T \}$, $p_T \in (0, 1]$, $c \in (0, 1]$, $\alpha > 0$, and $\beta > \ln\left( 2 (1 + c p_T e^{-\alpha})/(1 - e^{-\alpha}) \right)$. Let $D = \Ber(c)$. Let 
$$S_\pi = \{ (1, \Delta), (2, \Delta), (0, 1), (1, 1), (0, 2) \}.$$ 
Let $\pi$ be an ordering on $S_\pi$ that is a subsequence of an optimal priority ordering. Suppose $\pi_0 = (0, 1)$, $\pi_1 = (0, 2)$, and $\pi_2 = (1, 1)$. Then $\pi_3 = (1, \Delta)$. 
\end{lemma}
\begin{proof}
To determine $\pi_3$, we compute the index function for the only remaining types, $(1, \Delta)$ and $(2, \Delta)$. Note that all the potential descendants of types $(1, \Delta)$ and $(2, \Delta)$ are in the prefix. By~\cref{lem:gamma_lemma}, since the type $(2, \Delta)$ has at least one potential child in the prefix, $\gamma((2, \Delta), \pi_2) \leq \gamma((1, \Delta), \pi_2)$. Therefore
\begin{align}
I((1, \Delta), \pi_2) - I((2, \Delta), \pi_2) &= \frac{ \Exp[b((1, \Delta), \pi_2)]}{1 - \gamma((1, \Delta), \pi_2)} - \frac{ \Exp[b((2, \Delta), \pi_2)]}{1 - \gamma((2, \Delta), \pi_2)} \nonumber \\
&\geq \frac{1}{1 - \gamma((1, \Delta), \pi_2)} \cdot \left( \Exp[b((1, \Delta), \pi_2)] - \Exp[b((2, \Delta), \pi_2)] \right) \label{eq:first_prefix_compare}
\end{align}
To compute $\Exp[b((1, \Delta), \pi_2)]$, we examine the $((1, \Delta), \pi_2)$-period. The period begins by querying the node of type $(1, \Delta)$ on step $t = 0$. With probability $p(\Delta)$ the node is infected, and with probability $c$ it has child $(0, 1)$. So with probability $p(\Delta) c$ the node of type $(0, 1)$ is queryied on step $t = 1$. Therefore
\begin{align}
\Exp[b((1, \Delta), \pi_2)] &= p(\Delta) \left( b(1) + c p(1) b(0) e^{-\beta} \right) \label{eq:r_1}
\end{align}
We follow a similar process to compute $\Exp[b((2, \Delta), \pi_2)]$. The $((2, \Delta), \pi_2)$-period begins by querying the node of type $(2, \Delta)$ on step $t = 0$, which is infected with probability $p(\Delta)$. Following the priority ordering defined by the prefix, if all potential descendants are realized, then the child of type $(0, 2)$ is queried on step $t = 1$, followed by the child of type $(1, 1)$ on step $t = 2$; if the child of type $(1, 1)$ is infected, then its child of type $(0, 1)$ is queried on step $t = 3$. Accounting for the fact that each potential descendant is realized independently with probability $c$, 
\begin{align}
\Exp[b((2, \Delta), \pi_2)] &= p(\Delta) \big[ b(2) + c \left( p(2) b(0) e^{-\beta} + c p(1) (b(1) e^{-2 \beta} + c p(1) b(0) e^{-3 \beta} \right) + \nonumber \\
&\phantom{{}={}} (1 - c) c p(1) \left( b(1) e^{-\beta} + c p(1) b(0) e^{-2 \beta} \right) \big] \label{eq:r_2}
\end{align}
Combining~\cref{eq:r_1} and~\cref{eq:r_2}, 
\begin{align*}
\Exp[b((1, \Delta), \pi_2)] - \Exp[b((2, \Delta), \pi_2)]  &=  p(\Delta) \left( b(1) + c p(1) b(0) e^{-\beta} \right) - \\
&\phantom{{}={}} p(\Delta) \big[ b(2) + c \left( p(2) b(0) e^{-\beta} + c p(1) (b(1) e^{-2 \beta} + c p(1) b(0) e^{-3 \beta}) \right) + \\
&\phantom{{}={}} (1 - c) c p(1) \left( b(1) e^{-\beta} + c p(1) b(0) e^{-2 \beta} \right) \big] \\
&=  p(\Delta) \big[ b(1) + c p(1) b(0) e^{-\beta} - \\
&\phantom{{}={}} b(2) - c p(2) b(0) e^{-\beta} - c^2 p(1) b(1) e^{-2 \beta} - c^3 p(1)^2 b(0) e^{-3 \beta} - \\
&\phantom{{}={}} (1 - c) c p(1) b(1) e^{-\beta} - (1 - c) c^2 p(1)^2 b(0) e^{-2 \beta} \big] \\
&=  p_T e^{-\alpha \Delta} \big[ e^{-\beta} + c p_T e^{-\alpha - \beta} - \\
&\phantom{{}={}} e^{-2 \beta} - c p_T e^{-2 \alpha - \beta} - c^2 p_T e^{-\alpha - 3 \beta} - c^3 p_T^2 e^{-2 \alpha - 3 \beta} - \\
&\phantom{{}={}} (1 - c) c p_T e^{-\alpha - 2 \beta} - (1 - c) c^2 p_T^2 e^{-2 \alpha - 2 \beta} \big] \\
&\geq p_T e^{-\alpha \Delta} \left( c p_T e^{-\alpha - \beta} - c p_T e^{-2 \alpha - \beta} - 2 c p_T e^{-\alpha - 2 \beta} - 2 c^2 p_T^2 e^{-2 \alpha - 2 \beta} \right)
\end{align*}
The above expression is strictly positive iff
\begin{align}
c p_T e^{-\alpha - \beta} &> c p_T e^{-2 \alpha - \beta} + 2 c p_T e^{-\alpha - 2 \beta} + 2 c^2 p_T^2 e^{-2 \alpha - 2 \beta} \label{eq:first_prefix_start} \\
1 &> e^{- \alpha } + 2 e^{- \beta} + 2 c p_T e^{- \alpha -  \beta} \\
1 &> e^{- \alpha } + 2 e^{- \beta}(1 + c p_T e^{-\alpha}) \\
1 - e^{-\alpha} &> 2 e^{- \beta}(1 + c p_T e^{-\alpha}) \\
e^{\beta} &> \frac{2 (1 + c p_T e^{-\alpha})}{1 - e^{-\alpha}} \\
\beta &> \ln\left( \frac{2 (1 + c p_T e^{-\alpha})}{1 - e^{-\alpha}} \right) \label{eq:first_prefix_end} 
\end{align}
Thus the inequality holds due to the bound on $\beta$, so $\Exp[r((1, \Delta), \pi_2)] - \Exp[r((2, \Delta), \pi_2)] > 0$. Therefore by~\cref{eq:first_prefix_compare}, $I((1, \Delta), \pi_2) > I((2, \Delta), \pi_2)$, so $\pi_3 = (1, \Delta)$.
\end{proof}

\begin{lemma} \label{lem:second_prefix}
Fix $T \in \NN$, $\Delta \in \{ 2, \dots, T \}$, $p_T \in (0, 1]$, $c \in (0, 1]$, $\alpha > 0$, and $\beta > \ln\left( 2 (1 + c p_T e^{-\alpha})/(1 - e^{-\alpha}) \right)$. Let $D = \Ber(c)$. Let 
$$S_\pi = \{ (1, \Delta), (2, \Delta), (0, 1), (1, 1), (0, 2) \}.$$ 
Let $\pi$ be an ordering on $S_\pi$ that is a subsequence of an optimal priority ordering. Suppose $\pi_0 = (0, 1)$ and  $\pi_1 = (1, 1)$. Then $I((1, \Delta), \pi_1) > I((2, \Delta), \pi_1)$. 
\end{lemma}
\begin{proof}
Since the type $(2, \Delta)$ has at least one potential child in the prefix, $\gamma((2, \Delta), \pi_1) \leq \gamma((1, \Delta), \pi_1)$. Therefore
\begin{align}
I((1, \Delta), \pi_1) - I((2, \Delta), \pi_1) &= \frac{ \Exp[b((1, \Delta), \pi_1)]}{1 - \gamma((1, \Delta), \pi_1)} - \frac{ \Exp[b((2, \Delta), \pi_1)]}{1 - \gamma((2, \Delta), \pi_1)} \nonumber \\
&\geq \frac{1}{1 - \gamma((1, \Delta), \pi_1)} \cdot \left( \Exp[b((1, \Delta), \pi_1)] - \Exp[b((2, \Delta), \pi_1)] \label{eq:second_prefix} \right)
\end{align}
The type $(1, \Delta)$ has one potential child $(0, 1)$ in the prefix and no other descendants, so
\begin{align}
\Exp[b((1, \Delta), \pi_1)] &= p(\Delta) [ b(1) + c p(1) b(0) e^{-\beta}] \label{eq:second_prefix_r1}
\end{align}
To compute $\Exp[b((2, \Delta), \pi_2)]$ we examine the $((2, \Delta), \pi_2)$-period. The period begins by querying the node of type $(2, \Delta)$ on step $t = 0$. With probability $p(\Delta)$ the node is infected, and with probability $c$ it has child $(1,1)$, which is then queried on step $t = 1$. With probability $p(\Delta)$ the node $(1, 1)$ is infected, and with probability $c$ it has child $(0, 1)$, which is then queried on step $t = 2$. Thus
\begin{align}
\Exp[b((2, \Delta), \pi_1)] &= p(\Delta) [ b(2) + c p(1) \left( b(1) e^{-\beta} + c p(1) b(0) e^{-2 \beta} \right)] \label{eq:second_prefix_r2}
\end{align}
Combining~\cref{eq:second_prefix_r1} and~\cref{eq:second_prefix_r2},
\begin{align*}
\Exp[b((1, \Delta), \pi_1)] - \Exp[b((2, \Delta), \pi_1)] &= p(\Delta) [ b(1) + c p(1) b(0) e^{-\beta} - b(2) - c p(1) \left( b(1) e^{-\beta} + c p(1) b(0) e^{-2 \beta} \right)] \\
&= p_T e^{-\alpha \Delta} [ e^{-\beta} + c p_T e^{-\alpha - \beta} - e^{-2 \beta} - c p_T e^{-\alpha - 2 \beta} - c^2 p_T^2 e^{-2 \alpha - 2 \beta}] \\
&\geq p_T e^{-\alpha \Delta} [c p_T e^{-\alpha - \beta} - c p_T e^{-\alpha - 2 \beta} - c^2 p_T^2 e^{-2 \alpha - 2 \beta}] \\
\end{align*}
The above expression is strictly positive iff
\begin{align*}
c p_T e^{-\alpha - \beta} &> c p_T e^{-\alpha - 2 \beta} + c^2 p_T^2 e^{-2 \alpha - 2 \beta} \\
1 &> e^{- \beta} + c p_T e^{- \alpha -  \beta} \\
1 &> e^{- \beta}(1 + c p_T e^{- \alpha}) \\
e^{\beta} &> 1 + c p_T e^{- \alpha} \\
\beta &> \ln(1 + c p_T e^{-\alpha})
\end{align*}
The inequality holds due to the bound on $\beta$, so $\Exp[b((1, \Delta), \pi_1)] - \Exp[b((2, \Delta), \pi_1)] > 0$. Therefore by~\cref{eq:second_prefix}, $I((1, \Delta), \pi_1) > I((2, \Delta), \pi_1)$.
\end{proof}
The following lemma is nearly identical to~\cref{lem:first_prefix} and references the prior lemma for many computations. The key difference between the two settings is that in the prefix that~\cref{lem:third_prefix} considers, $(1, 1) \prec (0, 2)$, whereas the opposite holds in~\cref{lem:first_prefix}. This difference in prefix changes the computation of the index function for $(2, \Delta)$.
\begin{lemma} \label{lem:third_prefix}
Fix $T \in \NN$, $\Delta \in \{ 2, \dots, T \}$, $p_T \in (0, 1]$, $c \in (0, 1]$, $\alpha > 0$, and $\beta > \ln\left( 2 (1 + c p_T e^{-\alpha})/(1 - e^{-\alpha}) \right)$. Let 
$$S_\pi = \{ (1, \Delta), (2, \Delta), (0, 1), (1, 1), (0, 2) \}.$$ 
Let $\pi$ be an ordering on $S_\pi$ that is a subsequence of an optimal priority ordering. Suppose $\pi_0 = (0, 1)$, $\pi_1 = (1, 1)$, and $\pi_2 = (0, 2)$. Then $\pi_3 = (1, \Delta)$. 
\end{lemma}
\begin{proof}
To select $\pi_3$, we compute the index function for the two remaining types, $(1, \Delta)$ and $(2, \Delta)$. Following the same argument as in~\cref{lem:first_prefix}, this reduces to comparing expected benefits over an $((h, \Delta), \pi_2)$-period. Therefore
\begin{align}
I((1, \Delta), \pi_2) - I((2, \Delta), \pi_2) &= \frac{ \Exp[b((1, \Delta), \pi_2)]}{1 - \gamma((1, \Delta), \pi_2)} - \frac{ \Exp[b((2, \Delta), \pi_2)]}{1 - \gamma((2, \Delta), \pi_2)} \nonumber \\
&\geq \frac{1}{1 - \gamma((1, \Delta), \pi_2)} \cdot \left( \Exp[b((1, \Delta), \pi_2)] - \Exp[b((2, \Delta), \pi_2)] \right) \label{eq:third_prefix_compare}
\end{align}
Just as in~\cref{lem:first_prefix}, the prefix includes $(0, 1)$, so the computation of $\Exp[b((1, \Delta), \pi_2)]$ is also the same.
\begin{align}
\Exp[b((1, \Delta), \pi_2)] &= p(\Delta) \left( b(1) + c p(1) r(0) e^{-\beta} \right) \label{eq:third_prefix_r1}
\end{align}
To compute $\Exp[b((2, \Delta), \pi_2)]$ we examine the $((2, \Delta), \pi_2)$-period. The period begins by querying the node of type $(2, \Delta)$ on step $t = 0$, which is infected with probability $p(\Delta)$. Following the priority ordering defined by the prefix, if all potential descendants are realized, then the child of type $(1, 1)$ is queried on step $t = 1$; if it is infected, then the node of type $(0, 1)$ is queried on step $t = 2$, followed by the child of type $(0, 2)$ on step $t = 3$. If the child of type $(1, 1)$ is not infected, then the node of type $(0, 1)$ is skipped, and the child of type $(0, 2)$ is queried on step $t = 2$. Accounting for the fact that each potential descendant is realized independently with probability $c$, 
\begin{align}
\Exp[b((2, \Delta), \pi_2)] &= p(\Delta) \Big[ b(2) + c \big[ p(1) \left( b(1) e^{-\beta} + c p(1) b(0) e^{-2 \beta} + c p(2) b(0) e^{-3 \beta} \right) + \nonumber \\
&\phantom{{}={}} (1 - p(1)) c p(2) b(0) e^{-2 \beta} \big] + (1 - c) c p(2) b(0) e^{-\beta} \Big] \label{eq:third_prefix_r2}
\end{align}
Combining~\cref{eq:third_prefix_r1} and~\cref{eq:third_prefix_r2},
\begin{align*}
\Exp[b((1, \Delta), \pi_2)] - \Exp[b((2, \Delta), \pi_2)] &= p(\Delta) \left( b(1) + c p(1) b(0) e^{-\beta} \right) - \\
&\phantom{{}={}} p(\Delta) \Big[ b(2) + c \big[ p(1) \left( b(1) e^{-\beta} + c p(1) b(0) e^{-2 \beta} + c p(2) b(0) e^{-3 \beta} \right) + \\
&\phantom{{}={}} (1 - p(1)) c p(2) b(0) e^{-2 \beta} \big] + (1 - c) c p(2) b(0) e^{-\beta} \Big] \\
&= p_T e^{-\alpha \Delta} \Big[ e^{-\beta} + c p_T e^{-\alpha - \beta} - \\
&\phantom{{}={}} e^{-2 \beta} - c p_T e^{-\alpha - 2 \beta} - c^2 p_T^2 e^{-2 \alpha - 2 \beta} - c^2 p_T^2 e^{-3 \alpha - 3 \beta} - \\
&\phantom{{}={}} (1 - p_T e^{-\alpha}) c^2 p_T e^{-2 \alpha - 2 \beta} - (1 - c) c p_T e^{-2 \alpha - \beta} \Big]\\
&\geq p_T e^{-\alpha \Delta} \Big[ c p_T e^{-\alpha - \beta} - (1 - c) c p_T e^{-2 \alpha - \beta} - \\
&\phantom{{}={}} \left( c p_T e^{-\alpha - 2 \beta} + (1 - p_T e^{-\alpha}) c^2 p_T e^{-2 \alpha - 2 \beta} \right) - \\
&\phantom{{}={}} \left( c^2 p_T^2 e^{-2 \alpha - 2 \beta} + c^2 p_T^2 e^{-3 \alpha - 3 \beta} \right) \Big] \\
&\geq p_T e^{-\alpha \Delta} \Big[ c p_T e^{-\alpha - \beta} - c p_T e^{-2 \alpha - \beta} - 2 c p_T e^{-\alpha - 2 \beta} - 2 c^2 p_T^2 e^{-2 \alpha - 2 \beta} \Big]
\end{align*}
By~\crefrange{eq:first_prefix_start}{eq:first_prefix_end}, the above expression is strictly positive, so $\Exp[b((1, \Delta), \pi_2)] - \Exp[b((2, \Delta), \pi_2)] > 0$. Therefore by~\cref{eq:third_prefix_compare} $I((1, \Delta), \pi_2) > I((2, \Delta), \pi_2)$, so $\pi_3 = (1, \Delta)$.
\end{proof}